%% file: ms.tex
\documentclass{article}

\usepackage{arxiv}

\usepackage[utf8]{inputenc} 
\usepackage[T1]{fontenc}    
\usepackage{hyperref}       
\usepackage{url}            
\usepackage{booktabs}       
\usepackage{amsfonts, amsmath, amsthm}       
\usepackage{nicefrac}       
\usepackage{microtype}      
\usepackage{lipsum}
\usepackage{fancyhdr}       
\usepackage{graphicx}       
\graphicspath{{media/}}     
\usepackage[linesnumbered,ruled,vlined]{algorithm2e}
\usepackage{tikz}
\usepackage{bm}
\newcommand{\abs}[1]{\lvert #1 \rvert}

\newtheorem{theorem}{Theorem}
\newtheorem{lemma}[theorem]{Lemma}
\newtheorem{definition}[theorem]{Definition}

\title{Differential Privacy on Dynamic Data}

\author{
  Yuan Qiu \\
  HKUST \\
  \texttt{yqiuac@cse.ust.hk} \\
   \And
  Ke Yi \\
  HKUST \\
  \texttt{yike@cse.ust.hk} \\
}

\begin{document}
\maketitle

\begin{abstract}
A fundamental problem in differential privacy is to release a privatized data structure over a dataset that can be used to answer a class of linear queries with small errors.  This problem has been well studied in the static case.  In this paper, we consider the dynamic setting where items may be inserted into or deleted from the dataset over time, and we need to continually release data structures so that queries can be answered at any time.  
We present black-box constructions of such dynamic differentially private mechanisms from static ones with only a  polylogarithmic degradation in the utility.   For the fully-dynamic case, this is the first such result.  For the insertion-only case, similar constructions are known, but we improve them over sparse update streams.
\end{abstract}

\keywords{Differential privacy\and Dynamic data\and Linear query}

\section{Introduction}

\subsection{Differential Privacy for Static Data}
Let $\mathcal{X}$ be a domain of items.  A dataset is a multiset of items $D\in \mathbb{N}^\mathcal{X}$.  Two datasets $D,D' \in \mathbb{N}^\mathcal{X}$ are \textit{neighbors}, denoted $D\sim D'$, if there exists an item $x\in \mathcal{X}$, such that $D=D'\cup\{x\}$ or vise versa\footnote{In this paper, all set operators on multisets denote their multiset versions.}.  \textit{Differential privacy (DP)}~\cite{DBLP:journals/fttcs/DworkR14} is defined as follows.

\begin{definition}[Differential Privacy~\cite{DBLP:journals/fttcs/DworkR14}]
A randomized mechanism $\mathcal{M}:\mathbb{N}^\mathcal{X}\to \mathcal{Y}$ satisfies $(\varepsilon,\delta)$-DP if for any neighboring datasets $D\sim D'$ and any subset of outputs $Y\subseteq \mathcal{Y}$, 
\begin{equation}\label{def:dp}
\Pr[\mathcal{M}(D) \in Y] \leq e^\varepsilon\cdot \Pr[\mathcal{M}(D') \in Y] + \delta\,.
\end{equation}
\end{definition}
Note that $\varepsilon$ is usually set to some constant, while $\delta$ should be negligible.  In particular, the $\delta=0$ case is referred to as \textit{pure-DP} or $\varepsilon$-DP, which provides a qualitative better privacy guarantee than the $\delta>0$ case. 

A \textit{linear query} is specified by a function $f:\mathcal{X}\to [0,1]$.  The result of evaluating $f$ on $D$ is defined as $f(D):=\sum_{x\in D} f(x)$.
A fundamental problem in differential privacy is the following: Given a set of linear queries $\mathcal{Q}=\{f_1,\dots,f_{\abs{\mathcal{Q}}}\}$, design a DP mechanism $\mathcal{M}$ that, on any given $D$, outputs a data structure $\mathcal{M}(D)$, from which an approximate $f(D)$ can be extracted for any $f\in \mathcal{Q}$.  Let $\mathcal{M}_f(D)$ be the extracted answer for $f(D)$. We say that $\mathcal{M}$ has error $\alpha$ with probability $1-\beta$, if 
\[
\Pr\left[\max_{f\in\mathcal{Q}}\left|\mathcal{M}_f(D)-f(D)\right|> \alpha\right] \leq \beta
\]
for any $D$, 
where the probability is taken over the internal randomness of $\mathcal{M}$.  
Clearly, the error $\alpha$ is a function of both the privacy budget $(\varepsilon, \delta)$ and the failure probability $\beta$.  For most mechanisms, it also depends on the data size $\abs{D}$, number of queries $\abs{\mathcal{Q}}$, and domain size $\abs{\mathcal{X}}$.  To simplify notation, we often omit some of these parameters from  the full list $\alpha(\varepsilon,\delta,\beta,\abs{D},\abs{\mathcal{Q}},\abs{\mathcal{X}})$ if it is clear from the context.

There is extensive work on the best achievable $\alpha$ for various families of linear queries.  This paper takes a black-box approach, i.e., we present dynamic algorithms that can work with any $\mathcal{M}$ that has been designed for a static dataset $D$.  The error for the dynamic algorithm will be stated in terms of the $\alpha$ function of the mechanism $\mathcal{M}$ that is plugged into the black box.  Nevertheless, we often derive the explicit bounds for the following two most interesting and extreme cases:

\medskip
\noindent\textbf{Basic counting.} If $\mathcal{Q}$ consists of a single query $f(\cdot)\equiv 1$, which simply returns $f(D)=\abs{D}$, then the ``data structure'' $\mathcal{M}(D)$ consists of just one number, which is a noise-masked $f(D)$.  The most popular choice of the noise is a random variable drawn from the Laplace distribution $\operatorname{Lap}(\frac{1}{\varepsilon})$, and the resulting mechanism satisfies $\varepsilon$-DP.
Its error is $\alpha_{\operatorname{Lap}}(\varepsilon,\beta)=O(\frac{1}{\varepsilon}\log{\frac{1}{\beta}})$.
Alternatively, one can add a Gaussian noise, which is $(\varepsilon, \delta)$-DP for $\delta>0$ and yields $\alpha_{\operatorname{Gauss}}(\varepsilon,\delta,\beta)=O\left(\frac{1}{\varepsilon}\sqrt{\log{\frac{1}{\delta}}\log{\frac{1}{\beta}}}\right)$.  The two error bounds are generally incomparable, but the former is usually better since $\delta\le \beta$ in common parameter regimes.

\medskip
\noindent\textbf{Arbitrary queries.} If $\mathcal{Q}$ consists of arbitrary linear queries, then the \textit{private multiplicative weights (PMW)}~\cite{DBLP:conf/focs/HardtR10,DBLP:conf/nips/HardtLM12} mechanism achieves 
\[
\alpha_{\operatorname{PMW}}(\varepsilon,\delta,\beta,\abs{D},\abs{\mathcal{Q}},\abs{\mathcal{X}})=\left\{
\begin{aligned}
&O\left(\abs{D}^\frac{2}{3}\left(\frac{\log\abs{\mathcal{X}} \log (\abs{\mathcal{Q}}/\beta)}{\varepsilon}\right)^\frac{1}{3}\right), & \delta = 0; \\
&O\left(\abs{D}^\frac{1}{2}\left(\frac{\sqrt{\log \abs{\mathcal{X}}\log(1/\delta)} \log (\abs{\mathcal{Q}}/\beta)}{\varepsilon}\right)^\frac{1}{2}\right), & \delta > 0.
\end{aligned}
\right.
\]
It is known that for $\abs{\mathcal{Q}}$ sufficiently large, PMW achieves the optimal error up to polylogarithmic factors.

There are many possibilities between the two extreme cases.  In each case, the achievable error bound $\alpha$ intricately depends on the discrepancy of the query set $\mathcal{Q}$, in addition to the aforementioned parameters.  We include a brief review in Appendix \ref{app:review}, which is not necessary for the understanding of this paper. We make a reasonable assumption that $\alpha$ does not depend on any of those parameters exponentially, which allows us to ignore the constant coefficients in the parameters when writing big-O results.  This assumption holds for most existing mechanisms for linear queries, except for the so-called low-privacy regime $\varepsilon > \omega(1)$.  For results related to PMW, we often use the $\tilde{O}$ notation to further suppress dependencies on $\varepsilon$ and the polylogarithmic factors.

\subsection{Differential Privacy for Insertion-Only Streams}
Moving from static data to dynamic data, the simplest model is the insertion-only case,  which has been studied under the name \textit{differential privacy under continual observation} \cite {DBLP:conf/stoc/DworkNPR10,DBLP:conf/icalp/ChanSS10}.  In this model, time is divided into discrete units and the input is an insertion-only stream $\bm{x}=(x_1,\dots,x_T)$, where each $x_i \in \mathcal{X}\cup\{\bot\}$ and $T$ is possibly $\infty$. 
If $x_i=\bot$, no item arrives at time $i$. If $x_i\in\mathcal{X}$, item $x_i$ is inserted into the underlying dataset at time $i$.  The dataset at time $t\in \mathbb{N}$ is thus $D_t:=\cup_{i\leq t:x_i\neq \bot} \{x_i\}$. Given a set of linear queries $\mathcal{Q}$, the problem is to release an $\mathcal{M}^{(t)}(D_t)$ immediately after every time step $t$ (some $\mathcal{M}^{(t)}(D_t)$'s may be empty), so that (1) all the released data structures $(\mathcal{M}^{(1)}(D_1), \mathcal{M}^{(2)}(D_2), \dots)$ jointly satisfy $(\varepsilon, \delta)$-DP, and (2) an approximate $f(D_t)$ can be extracted from $(\mathcal{M}^{(1)}(D_1), \dots, \mathcal{M}^{(t)}(D_t))$ for any $f\in \mathcal{Q}$ and any $t$. 

Two insertion-only streams $\bm{x},\bm{x}'\in (\mathcal{X}\cup\{\bot\})^*$ are considered as neighbors, denoted $\bm{x}\sim \bm{x}'$, if they differ by one timestamp \cite{DBLP:conf/stoc/DworkNPR10, DBLP:conf/icalp/ChanSS10, DBLP:conf/asiacrypt/DworkNRR15}.  Without loss of generality, we may assume that for this only different timestamp $i$, one of $x_i$ or $x'_i$ is $\bot$.  This is commonly called the \textit{add-one/remove-one} policy. To contrast, the \textit{change-one} policy requires $x_i\ne x'_i$ but neither is $\bot$. The former is more general, since a neighboring pair under the change-one policy is a neighboring pair of distance $2$ under the add-one/remove-one policy, thus an $(\varepsilon, \delta)$-DP mechanism by our definition is $(2\varepsilon, 2e^{2\varepsilon}\cdot \delta)$-DP under change-one policy by \textit{group privacy} \cite{DBLP:books/sp/17/Vadhan17}.

We introduce some extra notation here.
For any time range $[a,b]\subseteq\mathbb{N}$, define $D([a,b]):=\cup_{a\leq i \leq b:x_i\neq \bot}\{x_i\}$.  Hence $D_t=D([1,t])$. Accordingly, $f(D([a,b]))$ will be abbreviated to $f([a,b])$.
We use $N_t$ to denote the total number of items inserted up until time $t$, i.e., $N_t = \sum_{i=1}^t \mathbf{1}[x_i\neq \bot]$ where $\mathbf{1}[\cdot]$ is the indicator function.
Let $n_t=\abs{D_t}$ be the size of the dataset at time $t$.
For insertion-only streams, we have $N_t=n_t$.

One important property of linear queries is that they are \textit{union-preserving}, i.e.,  $f(D^{(1)}\cup D^{(2)}) = f(D^{(1)}) + f(D^{(2)})$ for any $D^{(1)},D^{(2)}\in \mathbb{N}^\mathcal{X}$. Thus, a common technique for insertion-only streams is to divide $D$ into disjoint subsets $D^{(1)},D^{(2)},\dots, D^{(k)}$, compute an $\mathcal{M}(D^{(j)})$ for each $D^{(j)}$, and return $\mathcal{M}_f(D^{(1)}) + \cdots + \mathcal{M}_f(D^{(k)})$ as an approximation of $f(D)$. The total error is thus at most $k\cdot \alpha(\varepsilon, \delta, \frac{\beta}{k})$ by a union bound.  This can often be improved by exploiting certain properties of $\mathcal{M}$.  For example, if we use the Laplace mechanism for the basic counting problem, then by Bernstein’s inequality~\cite{book:HDP} the error can be improved from $k\cdot \alpha_{\operatorname{Lap}}(\varepsilon, \frac{\beta}{k}) = O(\frac{k}{\varepsilon}\log\frac{k}{\beta})$ to 
$\alpha_{\operatorname{Lap}}^{(k)}(\varepsilon,\beta)=O(\frac{1}{\varepsilon}(\sqrt{k\log\frac{1}{\beta}} + \log\frac{1}{\beta}))$; similar improvements are also possible for many other mechanisms, which we review in Appendix \ref{app:review}.  However, for arbitrary queries with PMW, the simple union bound remains the best known.  Again, to hide all these details into the black box, we use $\alpha^{(k)}$ to denote the error bound under such a disjoint union.  More formally, we assume that the static mechanism $\mathcal{M}$ for queries $\mathcal{Q}$ is equipped with error functions $\alpha^{(k)}$ for all $k\in \mathbb{N}$, such that for any datasets $D^{(1)},\dots,D^{(k)}\in \mathbb{N}^\mathcal{X}$, we have 
\[
\Pr\left[\max_{f\in\mathcal{Q}} \left|\sum_{j=1}^k\mathcal{M}_f(D^{(j)}) - f(\cup_{j=1}^k D^{(j)})\right|> \alpha^{(k)}\right] \le \beta,
\]
where the probability is taken over the independent internal randomness of $\mathcal{M}(D^{(j)}), j=1,\dots, k$.  Likewise, $\alpha^{(k)}$ is a function of $\varepsilon,\delta,\beta,\abs{\mathcal{Q}}, \abs{\mathcal{X}}$, as well as the total size of the datasets $|D|=\sum_{j=1}^k\abs{D^{(j)}}$.  In particular, we have $\alpha^{(1)}=\alpha$.

\begin{table*}[ht]
\begin{center}
\begin{tabular}{ c l l l} 
 \hline
 & General error bound & Basic counting (constant $\beta$) & Stream model\\
 \hline
\cite{DBLP:conf/stoc/DworkNPR10}
& $\alpha^{(\log T)}\left(\frac{\varepsilon}{\log T},\frac{\delta}{\log T}\right)$
& $\frac{1}{\varepsilon}\log^{1.5}{T}$
& Finite stream\\
\cite{DBLP:conf/icalp/ChanSS10}
& $\alpha^{(\log t)}\left(\frac{\varepsilon}{\log t}, \frac{\delta}{\log t}\right)$
& $\frac{1}{\varepsilon}\log^{1.5}{t}$
& Infinite stream\\
\cite{DBLP:conf/asiacrypt/DworkNRR15}
& $\alpha^{(\log n_t)}\left(\frac{\varepsilon}{\log n_t}, \frac{\delta}{\log n_t} \right) + \frac{1}{\varepsilon} \log\frac{T}{\beta}$
& $\frac{1}{\varepsilon}(\log^{1.5}{n_t} +\log{T})$
& Finite stream \\
New
& $\alpha^{(\log m_t)}\left(\frac{\varepsilon}{\log m_t}, \frac{\delta}{\log m_t} \right)+\frac{1}{\varepsilon} \log\frac{t}{\beta} $
& $\frac{1}{\varepsilon}(\log^{1.5}{n_t} +\log{t})$
& Infinite stream\\
& where $m_t = n_t + \log\log t$ \\
\hline
\end{tabular}
\caption{Comparison of results over insertion-only streams.  These bounds hold for a single query at any time $t$. Over a finite stream, replacing $\beta$ with $\beta/T$  turns them into bounds that hold for all queries simultaneously.}\label{tab:compare}
\end{center}
\end{table*}

\medskip

Existing work on insertion-only streams~\cite {DBLP:conf/stoc/DworkNPR10, DBLP:conf/icalp/ChanSS10, DBLP:conf/asiacrypt/DworkNRR15} has only studied the basic counting problem.  However, all their algorithms are actually black-boxed, so they can be instantiated with any mechanism for other query classes.  Using the notation introduced above, we summarize their results, as well as our new result, in Table~\ref{tab:compare}.  Comparing the general error bounds, we see that \cite{DBLP:conf/icalp/ChanSS10} is better than \cite{DBLP:conf/stoc/DworkNPR10} since $t\le T$.  In particular, \cite{DBLP:conf/icalp/ChanSS10} works for an infinite $T$.  \cite{DBLP:conf/asiacrypt/DworkNRR15} is also better than \cite{DBLP:conf/stoc/DworkNPR10}, since $n_t\le T$ and $\alpha^{(\log T)}\left(\frac{\varepsilon}{\log T}, \frac{\delta}{\log T}\right) \ge \frac{1}{\varepsilon}\log\frac{T}{\beta}$ for any $\alpha$.  However, \cite{DBLP:conf/icalp/ChanSS10} and \cite{DBLP:conf/asiacrypt/DworkNRR15} are incomparable because there is no relationship between $(\log t)^{O(1)}$ and $\log T$ (e.g., it is $\log^{1.5} t$ vs.~$\log T$ for basic counting).  Our new result can be seen as achieving the best of both worlds:  It is better than   \cite{DBLP:conf/icalp/ChanSS10} because $m_t = n_t + \log\log t = O(t)$ and  $\alpha^{(\log t)}\left(\frac{\varepsilon}{\log t}, \frac{\delta}{\log t}\right) \ge \frac{1}{\varepsilon}\log\frac{t}{\beta}$ for any $\alpha$, and the improvement is more significant over sparse streams where $n_t \ll t$.  
Our new result is also better than \cite{DBLP:conf/asiacrypt/DworkNRR15}: First, our algorithm works for an infinite stream whereas \cite{DBLP:conf/asiacrypt/DworkNRR15} needs a finite $T$; even for the finite stream case, our second term $\frac{1}{\varepsilon} \log\frac{t}{\beta}$ is better than the $\frac{1}{\varepsilon} \log\frac{T}{\beta}$ term of \cite{DBLP:conf/asiacrypt/DworkNRR15}.  For the first term, ours matches that of \cite{DBLP:conf/asiacrypt/DworkNRR15} for $n_t>\log \log t$; if $n_t<\log\log t$, the problem would be trivial as simply answering $0$ for all queries has error at most $n_t = O(\log\log t)$.  We also give the explicit bounds for the basic counting problem under constant $\beta$ in Table \ref{tab:compare}, using the Laplace mechanism as the black box.

All these algorithms can also support arbitrary queries by plugging in PMW, and we offer a similar improvement from $\log t$ to $\log(n_t + \log\log t)$.  Nevertheless, since PMW has a polynomial error $\tilde{O}(n_t^{1/2})$ (for $\delta>0$) or $\tilde{O}(n_t^{2/3})$ (for $\delta=0$), one usually does not pay much attention to the logarithmic factors, hence our improvement is minor in this case. Perhaps not realizing that the algorithm in \cite{DBLP:conf/icalp/ChanSS10} is black-boxed,  \cite{DBLP:conf/nips/CummingsKLT18} presented a dynamic, white-box version of PMW for the infinite insertion-only stream case, but the error is $\tilde{O}(n_t^{3/4})$ (for $\delta>0$). They also showed a black-box solution, but the error bound is also inferior to that of \cite{DBLP:conf/icalp/ChanSS10}: it is $\tilde{O}(n_t^{5/6})$ when instantiated with PMW.

\subsection{Differential Privacy for Fully-Dynamic Streams}
Building on our new algorithm for insertion-only streams, we next consider the fully-dynamic setting, which is the main technical result of this paper.  In the fully-dynamic model, the input is a stream in the form of $\bm{x}=((x_1,c_1),\dots,(x_T,c_T))$, where $x_i\in \mathcal{X}$, $c_i\in\{-1,1,\bot\}$, and $T$ may still be $\infty$.
At time $i$,
(1) if $c_i=\bot$, there is no update;
(2) if $c_i=1$, then $x_i\in\mathcal{X}$ is inserted into the dataset;
(3) if $c_i=-1$, then $x_i\in\mathcal{X}$ is deleted from the dataset ($x_i$ is assumed to exist in the dataset). 
Thus, the dataset at time $t$ is $D_t:=\cup_{i\leq t: c_i=1} \{x_i\} - \cup_{i\leq t: c_i=-1} \{x_i\}$.
Same as for the insertion-only setting, we wish to release an $\mathcal{M}^{(t)}(D_t)$ at each timestamp $t$ such that all releases are jointly differentially private, and those released up to time $t$ can be used to answer queries in $\mathcal{Q}$ on $D_t$.

Similar to the add-one/remove-one policy for insertion-only streams, we consider two fully-dynamic streams $\bm{x},\bm{x}' \in (\mathcal{X}\times\{-1,1,\bot\})^*$ to be neighbors, denoted $\bm{x}\sim \bm{x}'$, if one stream has one more update, which can either be an insertion or a deletion, than the other.  This also incorporates the other cases (up to a factor of $2$ in $\varepsilon$), e.g., the two streams insert different items at some timestamp, or one stream inserts an item while the other deletes an item.


We still use $N_t$ to denote the number of updates until time $t$, i.e., $N_t=\sum_{i=1}^t \mathbf{1}[c_i\neq \bot]$, and $n_t =\abs{D_t}$.  Unlike the insertion-only case, we have $n_t\ll N_t$ for fully-dynamic streams.  This is the most important difference between the insertion-only setting and the fully-dynamic setting.  In particular, $N_t$ always increases over time, while $n_t$ fluctuates and may even hit $0$.  Ideally, we would like the error on $D_t$ to depend on $\abs{D_t}=n_t$, not $N_t$.

The standard approach for the fully-dynamic case is to divide the update stream into two insertion-only steams: 
one only containing insertions and one only containing deletions but treating these deletions as insertions.
Let $D_t^{\operatorname{ins}}:=\cup_{i\leq t: c_i=1} \{x_i\}$ be all items inserted up to time $t$, and $D_t^{\operatorname{del}}:=\cup_{i\leq t: c_i=-1} \{x_i\}$ all items deleted up to time $t$.
Then  $D_t = D_t^{\operatorname{ins}}- D_t^{\operatorname{del}}$.
By the union preserving property, we have $f(D_t) = f(D_t^{\operatorname{ins}}) - f(D_t^{\operatorname{del}})$, so we can run two separate instances of the insertion-only algorithm.  Using our insertion-only algorithm above, the error would be $O\left(\alpha^{(\log m_t)}\left(\frac{\varepsilon}{\log m_t}, \frac{\delta}{\log m_t}, \beta, N_t \right)+\frac{1}{\varepsilon} \log\frac{t}{\beta} \right)$, where $m_t=N_t +\log\log t$.  Most importantly, the error depends on the number of updates $N_t$, not the data size $n_t$. For basic counting with constant $\beta$, this becomes $O(\frac{1}{\varepsilon}(\log^{1.5} N_t+ \log t))$, which we consider as satisfactory, although there is still a polylogarithmic dependency on $N_t$.  In fact, there is a lower bound of $\Omega(\log t) = \Omega(\log N_t)$ for insertion-only streams  \cite{DBLP:conf/stoc/DworkNPR10,bun2015differentially}.


Thus, we are more interested in making PMW (or any other mechanisms with polynomial errors, such as half-space queries \cite{DBLP:conf/stoc/MuthukrishnanN12}) fully dynamic, because the simple solution above, when instantiated with PMW, would result in an error of $\tilde{O}(N_t^{1/2})$ (for $\delta>0$) or $\tilde{O}(N_t^{2/3})$ (for $\delta=0$).  In this paper, we present a black-box fully-dynamic algorithm that achieves an error of $\alpha^{(\tilde{O}(1))}\left(\tilde{O}(\varepsilon), \tilde{O}(\delta), \beta, \abs{D}=n_t+ \tilde{O}(1)\right)$.  Plugging in PMW, this yields the optimal error of $\tilde{O}(n_t^{1/2})$ (for $\delta>0$) or $\tilde{O}(n_t^{2/3})$ (for $\delta=0$), up to polylogarithmic factors.

\section{Existing Work for Insertion-Only Streams}

In this section, we briefly review prior constructions for insertion-only streams, which will be useful for our algorithms as well. We first state two DP composition theorems.

\subsection{DP Composition Theorems}

\begin{theorem}[Sequential Composition~\cite{DBLP:journals/fttcs/DworkR14}]
Let $\mathcal{M}_i:\mathbb{N}^\mathcal{U}\to \mathcal{Y}_i$ each be an $(\varepsilon_i,\delta_i)$-DP mechanism.  Then the composed mechanism $\mathcal{M}(D)=(\mathcal{M}_1(D),\dots,\mathcal{M}_k(D))$ is $(\sum_{i=1}^k \varepsilon_i, \sum_{i=1}^k \delta_i)$-DP.
\end{theorem}
We note that there are many improved versions of sequential composition \cite{DBLP:conf/focs/DworkRV10,DBLP:journals/tit/KairouzOV17, DBLP:conf/tcc/BunS16} with better dependencies on $k$.  Nevertheless, as $k$ is logarithmic in all our constructions, these improved versions do not offer better bounds than the basic version above for $\delta$ negligible in $T$.

\begin{theorem}[Parallel Composition~\cite{DBLP:conf/sigmod/McSherry09}]
\label{thm:standard-parallel}
Let $\mathcal{U}=\mathcal{U}_1 \cup \cdots \cup \mathcal{U}_k$ be a partitioning of the universe $\mathcal{U}$, and let $\mathcal{M}_i:\mathbb{N}^{\mathcal{U}_i}\to \mathcal{Y}_i$ each be an $(\varepsilon_i,\delta_i)$-DP mechanism.  Then the composed mechanism $\mathcal{M}(D)=(\mathcal{M}_1(D\cap \mathcal{U}_1),\dots,\mathcal{M}_k(D\cap \mathcal{U}_k))$ is $(\max_{i=1}^k \varepsilon_i, \max_{i=1}^k \delta_i)$-DP.
\end{theorem}

It is possible to have $k=\infty$ for both sequential and parallel composition.  In this case, the $\sum$ and $\max$ would be replaced by $\lim \sum$ and $\sup$, respectively.

\subsection{Finite Stream, Dense Updates}

For a finite stream, i.e., $T$ is given in advance, we can build a binary tree on the timestamps $\{1,\dots, T\}$. This tree corresponds to a dyadic decomposition of the time domain, where each tree node $v$ is associated with a dyadic interval, which we denote by $D^{(v)}$.
Let $\mathcal{M}$ be a static DP mechanism with error function $\alpha^{(k)}$. 
The binary tree mechanism in \cite{DBLP:conf/stoc/DworkNPR10} releases $\mathcal{M}(D^{(v)})$ at the end of the interval for every node $v$ in the tree.  All nodes on the same level enjoy parallel composition\footnote{When using parallel composition here, the universe $\mathcal{U}$ in Theorem \ref{thm:standard-parallel} is different from the input domain $\mathcal{X}$ defined for the linear queries $\mathcal{Q}$.  More precisely, here we apply Theorem \ref{thm:standard-parallel} over the universe $\mathcal{U}=\mathcal{X} \times [T]$, and the partitioning is on $[T]$.}, while nodes from different levels use sequential composition.
Since any timestamp is contained by at most $\log T$ dyadic intervals, it suffices to run an $(\frac{\varepsilon}{\log T}, \frac{\delta}{\log T})$-DP mechanism at each node $v$ to guarantee $(\varepsilon,\delta)$-DP of the whole mechanism.  As each $D_t$ can be covered by at most $\log T$ disjoint time ranges\footnote{Technically, this is $\log t$, so a $\sqrt{\log T}$ in the bound can be improved to $\sqrt{\log t}$, but this is minor, and this result will be subsumed in the next subsection anyway.}, we obtain the error bound in the first row of Table \ref{tab:compare}.
Recent work~\cite{DBLP:conf/soda/HenzingerUU23,DBLP:conf/icml/Fichtenberger23} further improved the constant factors using the matrix mechanism~\cite{DBLP:journals/vldb/LiMHMR15}.

\subsection{Infinite Stream, Dense Updates}
The binary tree mechanism \cite{DBLP:conf/stoc/DworkNPR10} requires $T$ to be known in advance so that the privacy budgets $\varepsilon$ and $\delta$ can be divided appropriately.   For an infinite $T$, Chan et al.~\cite{DBLP:conf/icalp/ChanSS10} presented a clever construction that gets around the issue.  It divides the infinite stream into disjoint time ranges of exponentially growing sizes: $[1, 2), [2, 4), [4, 8), [8, 16), \dots$, and releases $\mathcal{M}(\mathcal{R})$ for each such time range $\mathcal{R}$, each of which uses privacy budget $(\frac{\varepsilon}{2}, \frac{\delta}{2})$. By parallel composition, all of them are jointly $(\frac{\varepsilon}{2}, \frac{\delta}{2})$-DP.  These releases allow us to answer queries on $D_t$ for $t=2^i$, and the error is 
$\alpha^{(\log t)}(\frac{\varepsilon}{2}, \frac{\delta}{2})$.  For $2^i <t<2^{i+1}$, $f(D_t)$ can be partitioned into $f([1,t))=f([1,2^i)) + f([2^i, t))$.  Then, an instance of the binary tree mechanism is run on $D([2^i,2^{i+1}))$ for each $i$ with a finite $T=2^i$ and privacy budget $(\frac{\varepsilon}{2}, \frac{\delta}{2})$.  Parallel composition also applies here.  The total error is dominated by that of the binary tree mechanism, but now $T\le t$ for each instance.  This gives us the second row in Table \ref{tab:compare}.

It should be clear that this mechanism can also handle the case where multiple insertions arrive at the same timestamp, which will be needed in the next subsection.

\subsection{Finite Stream, Sparse Updates}
Dwork et al.~\cite{DBLP:conf/asiacrypt/DworkNRR15} investigated the case where the update stream is sparse, i.e., $n_t \ll t$, which is motivated by many real-time applications that tend to use very small intervals between timestamps.  At the heart of their algorithm is the \textit{private partitioning} mechanism, which in an online fashion divides the update stream into segments such that with probability at least $1-\beta$, (A) each segment contains $O(\frac{1}{\varepsilon}\log\frac{T}{\beta})$ insertions, and (B) at most $O(n_t)$ segments are produced by time $t$.  Then, one can feed each segment as a ``super timestamp'' with multiple insertions to a dense-update algorithm, e.g., the mechanism \cite{DBLP:conf/icalp/ChanSS10} described above\footnote{In their paper \cite{DBLP:conf/asiacrypt/DworkNRR15}, they applied the binary tree mechanism on the segments, which results in an inferior bound where $n_t$ is replaced by an upper bound on $n_t$ given in advance. }.  Property (A) induces an additive error of $O(\frac{1}{\varepsilon}\log\frac{T}{\beta})$, while property (B) implies that the $t$ can be replaced by $n_t$ in the error bound.  This yields the third row in Table \ref{tab:compare}.

We now describe the private partitioning mechanism~\cite{DBLP:conf/asiacrypt/DworkNRR15}. In essence, it iteratively invokes the \textit{sparse vector technique (SVT)} \cite{DBLP:journals/fttcs/DworkR14,DBLP:journals/pvldb/LyuSL17} on the update stream, which is shown in Algorithm \ref{alg:svt}.

\vspace{-2mm}

\begin{algorithm}[ht]
\caption{Sparse Vector Technique}
\label{alg:svt}
\KwIn{Dataset $D$, privacy budget $\varepsilon$, threshold $\theta$, and a (possibly infinite) sequence of linear queries $f_i$}
\KwOut{A privatized index $\tilde{i}$ of the first query that $f_i(D)$ is above $\theta$}

$\hat{\theta} \gets \theta+\operatorname{Lap}(\frac{2}{\varepsilon})$\;
\ForEach{$i\gets 1, 2,\dots$}{
\If{$f_i(D)+\operatorname{Lap}(\frac{4}{\varepsilon}) > \hat{\theta}$}{
\textbf{Output} $\tilde{i}=i$ and \textbf{Halt}\;
}
}
\end{algorithm}

\begin{theorem}[\cite{DBLP:journals/fttcs/DworkR14,DBLP:conf/pods/DongY23}]
\label{thm:SVT}
The SVT mechanism with parameter $\varepsilon$ satisfies $\varepsilon$-DP.
For any sequence of queries $f_1,\dots,f_T$ (where $T$ can be infinite) and threshold $\theta$, we have the following guarantees on its accuracy:
(1) if there exists $t_1\leq T$ such that $f_i(D) < \theta-\frac{8}{\varepsilon}\ln\frac{2t_1}{\beta}$ for all $i\leq t_1$, then with probability $1-\beta$, SVT does not halt before $t_1$; and (2) if there exists $t_2\leq T$ such that $f_{t_2}(D)\geq \theta+\frac{6}{\varepsilon}\ln\frac{2}{\beta}$, then with probability $1-\beta$, SVT halts for some $\tilde{i}\leq t_2$. Further, the output satisfies $f_{\tilde{i}}(D)\geq \theta-\frac{6}{\varepsilon}\ln\frac{2t_2}{\beta}$.
\end{theorem}

To apply SVT for private partitioning, the idea in \cite{DBLP:conf/asiacrypt/DworkNRR15} is to consider the update stream $\bm{x}$ as $D$, and ask the queries $f_i(\bm{x}) = \abs{\{j\le i : x_j \ne \bot\}}$ with threshold $\theta=\Theta(\frac{1}{\varepsilon}\log\frac{T}{\beta})$. 
It can then be shown that when SVT outputs an $\tilde{i}$, it must have seen $\Theta(\frac{1}{\varepsilon}\log\frac{T}{\beta})$ updates with probability $1-\beta$.  Then the current segment is closed and  another SVT instance is started, and the process repeats. 
While the utility of this private partitioning mechanism (i.e., property (A) and (B) above) follows easily from Theorem \ref{thm:SVT}, the proof of privacy is nontrivial.  It is tempting to simply apply parallel composition, since the SVT instances are applied on disjoint segments of the stream.  However, Theorem \ref{thm:standard-parallel} requires the partitioning of the universe (the timestamps in this case) to be given in advance.  In particular, the partitioning should be independent of the internal randomness of the mechanisms in the composition, but in this case, the partitioning is exactly the outputs of the mechanisms, so there is no independence.  As a result, \cite{DBLP:conf/asiacrypt/DworkNRR15}  proved the privacy of this private partitioning mechanism from scratch, without relying on the privacy of SVT.

\section{New Algorithm for Infinite Insertion-Only Streams}
The private partitioning mechanism of \cite{DBLP:conf/asiacrypt/DworkNRR15} only works for a finite $T$, as it invokes a series of SVT instances with $\theta=\Theta(\frac{1}{\varepsilon}\log\frac{T}{\beta})$.  The key component in our new algorithm is an adaptation of their private partitioning mechanism to an infinite $T$.  Then we feed the segments into the infinite stream algorithm of \cite{DBLP:conf/icalp/ChanSS10}.

\subsection{Adaptive Parallel Composition}
We first give a new and simpler proof for the privacy of the private partitioning mechanism of \cite{DBLP:conf/asiacrypt/DworkNRR15}, by developing an adaptive version of the parallel composition theorem.  This immediately proves the privacy of private partitioning from the privacy of SVT.  Other than private partitioning, we imagine that this adaptive parallel composition theorem could also be useful in other applications. 

We first extend the DP definition to a mechanism that, in addition to its original output, also \textit{declares} the sub-universe that it has queried on.

\begin{definition}
\label{def:DPdeclare}
Given a mechanism $\mathcal{M}:\mathbb{N}^\mathcal{U} \to \mathcal{Y} \times 2^\mathcal{U}$, we say that $\mathcal{M}$ is an $\varepsilon$-DP mechanism with declaration, if for any neighboring instances $D\overset{x}{\sim} D'\in \mathbb{N}^\mathcal{U}$ that differ by item $x\in\mathcal{U}$ and any output $(y, U)\in \mathcal{Y}\times 2^\mathcal{U}$ we have
\begin{alignat}{3}
\label{eq:xinU}
&\Pr[\mathcal{M}(D)=(y,U)]\leq e^\varepsilon\cdot &&\Pr[\mathcal{M}(D')=(y,U)]\,, &&\quad \text{if } x\in U\,; \\
\label{eq:xnotinU}
&\Pr[\mathcal{M}(D)=(y,U)]= &&\Pr[\mathcal{M}(D')=(y,U)]\,, &&\quad\text{if } x\not\in U\,.
\end{alignat}
\end{definition}

Note that an $\varepsilon$-DP mechanism with declaration is also $\varepsilon$-DP; in fact, Definition \ref{def:DPdeclare} imposes a stronger requirement \eqref{eq:xnotinU} for the $x\not \in U$ case than standard $\varepsilon$-DP.   Conversely, any standard $\varepsilon$-DP mechanism  $\mathcal{M}':\mathbb{N}^\mathcal{U}\to \mathcal{O}$ can be turned into an $\varepsilon$-DP mechanism with declaration $\mathcal{M}(D)=(\mathcal{M}'(D), \mathcal{U})$, which always declares the entire universe $\mathcal{U}$, but such a trivial declaration does not allow parallel composition.
We present adaptive parallel composition in Theorem~\ref{thm:parallel}.

\begin{theorem}[Adaptive Parallel Composition]\label{thm:parallel}
Let $\mathcal{M}_1,\dots,\mathcal{M}_k$ (where $k$ can be infinite) each be an $\varepsilon$-DP mechanism with declaration, where the mechanisms may be chosen adaptively (the choice of $\mathcal{M}_{i+1}$ may depend on the previous outputs $(y_1, U_1),\dots,(y_i, U_{i+1})$).
If the $U_i$'s declared by the mechanisms are always pairwise disjoint, then the composed mechanism $\mathcal{M}=(\mathcal{M}_1,\dots,\mathcal{M}_k)$ is $\varepsilon$-DP.
\end{theorem}

\begin{proof}
For any neighboring datasets $D\overset{x}{\sim} D'$ that differ by item $x\in\mathcal{U}$ and any sequence of outputs $(y_1, U_1),\dots, (y_k, U_k)$, let $U_i$ be the unique set containing $x$ (let $i=k+1$ if no such $U_i$).
\begin{alignat*}{4}
  &       && &&\Pr[\mathcal{M}(D)=(y_1, y_2, \dots, y_k, U_1, U_2,\dots U_k)] \\
= &       && \Pi_{j=1}^{i-1} &&\Pr[\mathcal{M}_j(D)=(y_j,U_j) \mid (y_1, \dots, y_{j-1}, U_1,\dots, U_{j-1})] \\
  & \cdot\, && && \Pr[\mathcal{M}_i(D)=(y_i,U_i) \mid (y_1, \dots, y_{i-1}, U_1,\dots, U_{i-1})] \\
  & \cdot\, && \Pi_{j=i+1}^{k} && \Pr[\mathcal{M}_j(D)=(y_j,U_j) \mid (y_1, \dots, y_{j-1}, U_1,\dots, U_{j-1})]\\
\leq &    && \Pi_{j=1}^{i-1} && \Pr[\mathcal{M}_j(D')=(y_j,U_j) \mid (y_1, \dots, y_{j-1}, U_1,\dots, U_{j-1})] \\
  & \cdot\, && e^\varepsilon \cdot && \Pr[\mathcal{M}_i(D')=(y_i,U_i) \mid (y_1, \dots, y_{i-1}, U_1,\dots, U_{i-1})] \\
  & \cdot\, && \Pi_{j=i+1}^{k} && \Pr[\mathcal{M}_j(D')=(y_j,U_j) \mid (y_1, \dots, y_{j-1}, U_1,\dots, U_{j-1})]\\
= & && e^\varepsilon \cdot && \Pr[\mathcal{M}(D')=(y_1, y_2, \dots, y_k, U_1, U_2,\dots U_k)]\,.
\end{alignat*}
For the inequality, observe that $x \not\in U_j$ for $j\ne i$, so we apply \eqref{eq:xnotinU} for $j\ne i$, and apply \eqref{eq:xinU} for $\mathcal{M}_i$.
\end{proof}

Applying SVT on the update stream $\bm{x}$, the universe is the timestamps $[T]$, and the mechanism has no output $y$ but only the declaration $U$, which is the segment it produces.  The next SVT instance is applied after the last segment, so the declarations are disjoint.   Thus, the privacy of the private partitioning mechanism follows immediately from Theorem \ref{thm:parallel}.  Furthermore, this also holds even if each SVT uses a different threshold $\theta$ chosen adaptively, which is exactly what we will do next in order to extend private partitioning to infinite streams.

\subsection{Private Partitioning for Infinite Streams}
\label{sec:partition}
The finite private partitioning mechanism uses SVT instances with $\theta=\Theta(\frac{1}{\varepsilon}\log\frac{T}{\beta})$.  To deal with an infinite $T$, our infinite private partitioning mechanism (Algorithm~\ref{alg:partition}) uses quadratically increasing values for $T$.  This allows us to bound the number of updates inside each segment by $O(\frac{1}{\varepsilon}\log\frac{t}{\beta})$, plus at most $O(\log\log t)$ extra segments by time $t$. In addition, to make sure that these guarantees hold simultaneously for infinitely many $t$'s by a union bound, we allocate $O(\frac{\beta}{j^2})$ failure probability to each SVT instance. We more formally prove these utility guarantees below.

\begin{algorithm}[ht]
\KwIn{Update stream $\bm{x}=(x_1,\dots,x_t,\dots)$, privacy budget $\varepsilon$, failure probability $\beta$}
\KwOut{Segments $s_1=[1,t_1], s_2=[t_1+1,t_2],\dots$}
Initialize $t_0\gets 0$, $j\gets 1$, $T_1\gets 2$, $\beta_1 \gets \frac{6}{\pi^2}\beta$, $\theta_1 \gets \frac{7}{\varepsilon}\ln\frac{2T_1 }{\beta_1}$\; 
Initiate an SVT instance with privacy budget $\varepsilon$ and threshold $\theta_1$\; 
\ForEach{$t\gets 1,2,\dots$}{
Ask the query $\abs{\{t_{j-1} < i\le t : x_i \ne \bot\}}$ to SVT\;
\If{SVT halts with output $t$ OR $t \geq T_j$}{
Close the current segment, i.e., \textbf{output} $t_j\gets t$\;
$j\gets j+1$, $T_j \gets t^2$, $\beta_j \gets \frac{6}{\pi^2j^2}\beta$, $\theta_j \gets \frac{7}{\varepsilon}\ln\frac{2 T_j }{\beta_j}$\; 
Initiate a new SVT instance from $x_{t+1}$ with privacy budget $\varepsilon$ and threshold $\theta_j$\;
}
}
\caption{Infinite Private Partitioning}\label{alg:partition}
\end{algorithm}

\begin{lemma}
Algorithm~\ref{alg:partition} is $\varepsilon$-DP. With probability at least $1-\beta$, the following holds for all $t$: (1) every segment produced before time $t$ contains $O(\frac{1}{\varepsilon}\log \frac{t}{\beta})$ insertions, and (2) $O(n_t+\log\log t)$ segments are produced by time $t$.
\end{lemma}

\begin{proof}
Privacy follows directly from Theorem \ref{thm:parallel}.  Below we prove the utility.
Consider the $j$-th segment, we have $s_j=[t_{j-1}+1, t_j]$ and $\theta_j=\frac{7}{\varepsilon}\ln\frac{2 T_j }{\beta_j}$.
We discuss on the number of updates received after time $t_{j-1}$ and before $T_j$, i.e.\ the size of $D([t_{j-1}+1, T_j])$.

If the number of updates $\abs{D([t_{j-1}+1, T_j])} \geq \theta_j+\frac{6}{\varepsilon}\ln\frac{2}{\beta_j}$, let $t^*\leq T_j$ be the timestamp of the $(\theta_j+\frac{6}{\varepsilon}\ln\frac{2}{\beta_j})$-th update. 
By Theorem~\ref{thm:SVT}, with probability $1-\beta_j$, SVT closes the segment at some $t_j\leq t^*\leq T_j$ and $\abs{D([t_{j-1}+1, t_j])} \geq \theta_j-\frac{6}{\varepsilon}\ln\frac{2t^*}{\beta_j}\geq \frac{1}{\varepsilon}\ln\frac{2T_j}{\beta_j}$.
In this case the segment contains  $\Theta(\frac{1}{\varepsilon}\log\frac{T_j}{\beta_j})$ updates, which also implies there can be at most $n_t$ such segments at time $t$ when $\beta$ is small.

Otherwise $\abs{D([t_{j-1}+1, T_j])} < \theta_j+\frac{6}{\varepsilon}\ln\frac{2}{\beta_j}$.
Since the segment is closed no later than $T_j$, we know $t_j\leq T_j$, so that $\abs{D([t_{j-1}+1,t_j])}\leq \abs{D([t_{j-1}+1, T_j])}=O(\frac{1}{\varepsilon}\log\frac{T_j}{\beta_j})$.
But we are not able to lower bound the number of updates within such a segment.
Yet whenever this happens, the next $T_{j+1} = T_j^2$, so this can happen at most $O(\log\log t)$ times up until time $t$.

In total, there can be at most $m_t=O(n_t + \log\log t)$ segments with high probability, where only $O(\log\log t)$ of them can be empty.
The final step is taking an union bound over an infinite sequence of fail probabilities.
We allocate $\beta_j=\frac{6}{\pi^2 j^2}\beta$ to get $\sum_{j=1}^\infty \beta_j=\beta$, 
Then with probability $1-\beta$ the number of items in segment $s_j=[t_{j-1}+1, t_j]$ is 
\[
O\left(\frac{1}{\varepsilon}\log \frac{T_j}{\beta_j}\right)=O\left(\frac{1}{\varepsilon}\log \left(t_j^2\cdot \frac{j^2}{\beta}\right)\right)=O\left(\frac{1}{\varepsilon}\log\frac{t}{\beta}\right)\,,
\]
for $t_j\leq t$ and $j\leq m_t=O(t)$.
\end{proof}

Feeding the segments to the mechanism of \cite{DBLP:conf/icalp/ChanSS10} yields the following result:
\begin{theorem}\label{thm:our-ins}
Suppose there is a static $(\varepsilon,\delta)$-DP mechanism for answering a class $\mathcal{Q}$ of linear queries with error function $\alpha^{(k)}(\varepsilon,\delta,\beta)$.  Then there is a dynamic $(\varepsilon,\delta)$-DP mechanism for an infinite insertion-only stream that answers every query in $\mathcal{Q}$ on $D_t$ for every $t$ with error $O\left(\alpha^{(\log m_t)}\left(\frac{\varepsilon}{\log m_t}, \frac{\delta}{\log m_t}, \beta\right)+\frac{1}{\varepsilon} \log\frac{t}{\beta}\right)$ with probability at least $1-\beta$, where $m_t=n_t +\log\log t$.
\end{theorem}

For continual counting, plug
$\alpha_{\operatorname{Lap}}^{(k)}(\varepsilon,\beta)=O\left(\frac{1}{\varepsilon}\left(\sqrt{k\log\frac{1}{\beta}} + \log\frac{1}{\beta}\right)\right)$
into Theorem~\ref{thm:our-ins}, we obtain the error
$
O\left(\frac{1}{\varepsilon}\left(
\log^{1.5}{m_t} \sqrt{\log{\frac{1}{\beta}}}
+ \log m_t\log\frac{1}{\beta}
+ \log{\frac{t}{\beta}}\right)\right)
$.
When $\beta$ is a constant, this simplifies to $O\left(\frac{\log^{1.5}{n_t}+\log t}{\varepsilon}\right)$.

\section{Fully-Dynamic Streams}
Our fully-dynamic algorithm consists of the following steps.  First, we run the private partitioning mechanism from Section \ref{sec:partition} to divide the update stream into segments.  This will, with high probability, produce $m_t=O(N_t+\log\log t)$ segments by time $t$ where each segment contains $O(\frac{1}{\varepsilon}\log\frac{t}{\beta})$ updates.  This effectively reduces the number of timestamps to $m_t$ while incurring an additive error of $O(\frac{1}{\varepsilon}\log\frac{t}{\beta})$.  For each timestamp, we process multiple updates in a batch. As in Figure~\ref{fig:batch}, we treat all updates arriving in segment $s_i=[t_{i-1}+1, t_i]$ as if they all arrive at time $t_i$. 

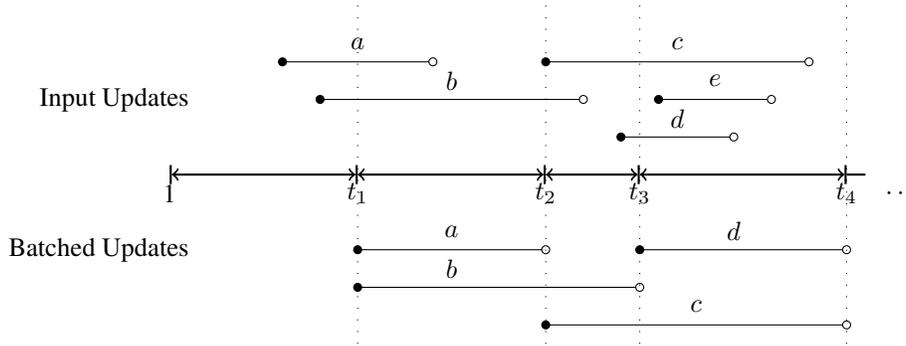
\begin{figure}[htbp]
\input{batch}
\centering
\caption{Multiple updates processed in a batch}\label{fig:batch}
\end{figure}

We can treat the update stream as a set of labeled intervals over the $m=m_t$ timestamps $t_1,t_2,\dots, t_m$.
An interval $[t_i,t_j)$ labeled with item $x\in \mathcal{X}$ represents an insertion-deletion pair $(x_{t_i}, 1),(x_{t_j},-1)$, where $x_{t_i}=x_{t_j}=x$.  We may assume that  $j>i$; if $i=j$ (e.g.,  interval $e$ in Figure~\ref{fig:batch}), the item is deleted within the same segment that it was inserted, this insertion-deletion pair is discarded.  Ignoring such pairs will only cause an additive error of at most $O(\frac{1}{\varepsilon}\log\frac{t}{\beta})$.  It is also possible that $t_j=\infty$, if the item is never deleted. Note that this interval representation of an update stream is not unique, e.g., when many copies of the same item are inserted and then deleted.  Any representation is fine; in fact, our algorithm does not depend this interval representation, only the analysis does.

Using the interval representation, a query on $D_t$ queries all intervals that are stabbed\footnote{Interval $[t_i, t_j)$ is stabbed by $t$ if $t_i\leq t < t_j$.} by $t$, so a natural idea is to use the \textit{interval tree}~\cite{DBLP:books/lib/BergCKO08} to organize these intervals.  In an interval tree, each $D_t$ is decomposed into a logarithmic number of subsets, each of which consists of one-sided intervals, which will allow us to use the insertion-only mechanism.  However, there are two technical difficulties in implementing such a plan.  First, the intervals are given in an online fashion, i.e., at time $t$, we only see the endpoints of the intervals prior to $t$. When we see the left endpoint of an interval, we do not know where in the interval tree to put this interval, yet, we need to immediately release privatized information about this interval.  Second, the interval tree on an infinite stream is also infinitely large, so we have to build it incrementally, while allocating the privacy budget appropriately.  We describe how to overcome these difficulties in Section \ref{sec:intervaltree}.
In Section \ref{sec:eachnode}, we introduce a DP mechanism  running at each node of the new tree structure to support querying at any time with respect to intervals stored in the tree.
The output of the whole mechanism is obtained by combining the individual mechanisms at tree nodes, which is summarized in Section  \ref{sec:together}.

\newcommand{\T}{\mathcal{T}}
\subsection{Online Interval Tree}
\label{sec:intervaltree}
We first build a binary tree $\T$ on the $m$ segments $s_1,\dots, s_m$ produced by the private partitioning mechanism. 
Since all the updates are batched, we can use segment $s_i=[t_{i-1}+1, t_i]$ and its right endpoint $t_i$ interchangeably.
Figure~\ref{fig:int-tree} shows an interval tree built on $8$ segments. It is clear that the $\T$ has $m$ nodes and at most $\log m$ height.  In the online setting, as the private partitioning mechanism produces more segments, $\T$ will also grow from left to right.  We order the nodes using an in-order traversal of $\T$: $v_1,v_2,\dots,$ and we build $v_i$ right after segment $s_i$ gets closed at time $t_i$ (see Figure \ref{fig:int-tree}). We also denote $t(v_i)=t_i$.

\begin{figure}[ht]
\centering
\begin{minipage}{.5\textwidth}
  \centering
  \input{int-tree}
  \caption{A standard interval tree}
  \label{fig:int-tree}
\end{minipage}%
\begin{minipage}{.5\textwidth}
  \centering
  \input{our-tree}
  \caption{An online interval tree}
  \label{fig:our-tree}
\end{minipage}
\end{figure}

We ignore differential privacy for now, and just focus on how to answer a stabbing query using an interval tree, i.e., report all intervals stabbed by a query.
In a standard interval tree, an interval is stored at the highest node $v$ such that $t(v)$ stabs the interval.  
We use $D(v)$ to denote the set of labeled intervals stored at $v$.  
For example $D(v_4)=\{a,c,d\}$ in Figure~\ref{fig:int-tree}.  For a query at time $t_{q}$, we follow the root-to-leaf path to $v_q$ in $\T$.  For each node $v$ on the path whose $t(v)\leq t_q$, we find all intervals in $D(v)$ whose right endpoints are on the right side of $t_q$; for each node $v$ on the path whose $t(v)>t_q$, we find all intervals in $D(v)$ whose left endpoints are on the left side of (or equal to) $t_q$.  Standard analysis on the interval tree shows that these subsets form a disjoint union of all intervals stabbed by $t_q$.
For example when a query arrives at time $t_5$ in Figure~\ref{fig:int-tree}, we follow the path $(v_5,v_6,v_4,v_8,\dots)$.  $t(v_4)\leq t_5$, where $a,c\in D(v_4)$ have their right endpoints on the right side of $t_5$; $t(v_6)>t_5$, where $e\in D(v_6)$ has its left endpoint equal to $t_5$.
Thus we report $D_{t_5}=\{a,c,e\}$, which are the elements present in the dataset at time $t_5$.

In an online setting, however, we do not know which node is the highest to put an interval in, since we do not know the deletion time when an item is inserted.  The idea is to put a copy of the interval into every node where the interval \textit{might} be placed into.
We use Figure~\ref{fig:our-tree} to illustrate. Interval $a$ will placed into $v_1,v_2,v_4$, as well as ancestors of $v_4$, while interval $e$ will be placed into $v_5,v_6,v_8$ (and its ancestors).  As there are infinitely many nodes where an interval might be placed into, we do not actually put an interval in all those nodes, but will do so lazily.  Thus, the rule is that an interval $[t_i, t_j)$ will be stored at $v_i$ and \textit{each} ancestor $\tilde{v}$ of $v_i$ where $t(\tilde{v})$ stabs $[t_i, t_j)$.
In Figure~\ref{fig:our-tree}, both $a$ and $b$ are inserted in segment $s_1$, so $D(v_1)$ stores both. The ancestors of $v_1$ are $v_2,v_4,v_8$ (and possibly more). $v_2$ stores $a$ but not $b$, since $t(v_2)$ only stabs interval $a$.
Intuitively, by time $t(v_2)$, $b$ is already deleted, so there is no need to store $b$ at $v_2$.  On the other hand, $a$ needs to be stored in both $v_1$ and $v_2$ (in the standard interval tree, it is only stored at $v_4$), because by time $t(v_1)$ or $t(v_2)$, we still do not know its deletion time.
Note that the extra copies of an interval are only stored in the ancestors of the node corresponding to its insertion time, thus there is at most one extra copy at each level. For example, although interval $a$ is stabbed by $t(v_3)$, it is not stored there since $v_3$ is not an ancestor of $v_1$.

\subsubsection{Building an Online Interval Tree}

This online interval tree can be incrementally constructed easily.  After the batch of updates in segment $s_i$ have arrived, we can construct $D(v_i)$.
For nodes in the left-most path ($v_1,v_2,v_4,v_8,\dots$), $D(v_i)=D_{t_i}$ simply consists of all items currently in the dataset, since the node $v_i$ is always an ancestor of $v_j$ for $j<i$.
For other nodes $v_i$, there exists at least one ancestor on its left, and let $\tilde{v}$ be the lowest such ancestor.
We include into $D(v_i)$ all labeled intervals in the current $D_{t_i}$ that is inserted after $t(\tilde{v})$.
For example when constructing $D(v_5)$ in Figure \ref{fig:our-tree}, we find this left-ancestor to be $\tilde{v}=v_4$, so $D(v_5)=\{e\}$ will only include $e$ from the current dataset $D_5=\{a,c,e\}$, which is inserted after $t_4$. Intuitively, $a$ and $c$ have already been covered by $v_4$.

Note that when $D(v_i)$ is first constructed, we do not have the deletion times of the items in $D(v_i)$, which will be added when these items are actually deleted later.  For example, in Figure \ref{fig:our-tree}, $D(v_1)=\{a, b\}$ is constructed after segments $s_1$ but neither item is associated with a deletion time.  After segment $s_2$, we add the deletion time of $b$, augmenting $D(v_1)$ to $D(v_1) = \{a, (b,t_2)\}$; after segment $s_8$, $D(v_1)$ becomes $\{(a,t_8), (b,t_2)\}$.  Note that there is no need to associate the left endpoints (i.e., insertion times) to the items as in the standard interval tree, and we will see why below.  





\subsubsection{Querying an Online Interval Tree}

Now we show how to answer a stabbing query using the online interval tree.  Since the online interval tree includes multiple copies of an item, the standard interval tree query algorithm will not work, as it may report duplicates.  For the stabbing problem itself, duplicates are not an issue as they can be easily removed if they have been reported already.  However, for answering linear queries, we actually need to cover all  stabbed intervals by a disjoint union of subsets.  To achieve it, we modify the stabbing query process as follows.  Suppose we ask a query at time $t_q$.  We first follow the root-to-leaf path to $v_q$ in $\T$.  For each node $\tilde{v}$ on the path whose $t(\tilde{v})\leq t_q$, we report all the items in $D(\tilde{v})$ whose deletion time is on the right side of $t_q$. Again consider a query at time $t_5$ in Figure~\ref{fig:our-tree}, we will only visit $v_4$ and $v_5$ who report $\{a,c\}$ and $\{e\}$ respectively.

Note that unlike in the standard interval tree, we do not query those nodes on the right side of $t_q$ (e.g.~$v_6$).  It turns out that the items stored in those nodes are exactly compensated by the extra copies of items stored in the nodes on the left side of $t_q$. The following lemma formalizes this guarantee.

\begin{lemma}
The query procedure described above reports each stabbed interval exactly once.
\end{lemma}

\begin{proof}
Given a query at time $t_q$, consider any interval $[t_i, t_j)$.
If $t_q$ does not stab the interval, it should not be reported. This happens when: (1) the item has been deleted at query time ($t_j\leq t_q$). As we only report an item whose deletion time is on the right side of $t_q$, the interval is filtered out; (2) the item has not arrived by query time ($t_i > t_q$). As we only visit nodes where $t(\tilde{v})\leq t_q$, it follows that $t(\tilde{v})\leq t_q <t_i$. By definition, $t(\tilde{v})$ does not stab $[t_i, t_j)$, thus does not store the inverval.

 \begin{figure}[ht]
\input{proof}
\centering
\caption{Query procedure for a stabbing query}\label{fig:proof}
\end{figure}
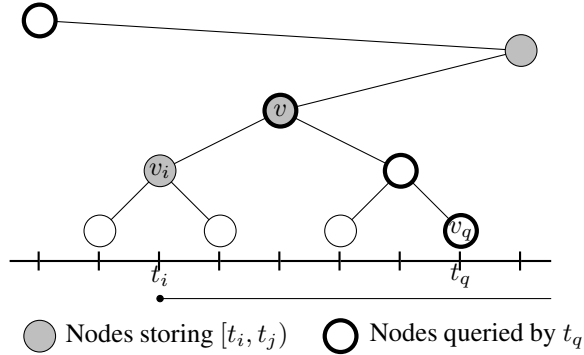

The final case is when $t_i\leq t_q < t_j$, and the interval should be reported by exactly one node. This is shown in Figure~\ref{fig:proof}.
For the trivial case that $q=i$, the newly constructed node $v_i$ is the only node reporting this interval.
Otherwise, consider the minimum subtree containing both $v_i$ and $v_q$. Assume it is rooted at $v$.
We must have $v_i$ in its left subtree and $v_q$ in its right subtree by the minimum property, with the only exception that one of them can be $v$ itself, i.e., $t_i\leq t(v) \leq t_q$.
We can argue that $v$ is the only node that reports the interval:
any node $v_i'\neq v$ that stores $[t_i,t_j)$ is either in the left subtree of $v$ (because it is an ancestor of $v_i$); or an ancestor of $v$ that is on the right side of $v$ (because $t(v_i')\geq t_i$);
any node $v_q'\neq v$ queried by $t_q$ is either in the right subtree of $v$ (because it is an ancestor of $v_q$); or an ancestor of $v$ that is on the left side of $v$ (because $t(v_q')\leq t_q$).
Thus the only node that can possibly report this interval is $v$.
Since $v$ is  queried by $v_q$, and $t_j>t_q$, this stabbing interval is reported exactly once by $v$.
\end{proof}

\subsection{Mechanism at Each Node}
\label{sec:eachnode}
We have shown that the online interval tree can be incrementally constructed, such that at any time $t_j$, we can obtain the current dataset $D_{t_j}$ by a disjoint union of $O(\log j)$ subsets, each from $v_j$ or a left-side ancestor of $v_j$ in the interval tree.  Consider each queried node $v_i$ ($i\leq j$), and let $D_{t}(v_i)$ be the set of items that node $v_i$ stores at time $t$.
This implies any linear query $f(D_t)$ can be answered by computing the sum $\sum_{i}f(D_t(v_i))$ over queried tree nodes.
Answering a query $f$ on items stored by $v=v_i$ at time $t=t_j$ is a deletion-only problem.
When $D(v)=D_{t_i}(v_i)$ is first constructed at time $t_i$, it consists of items in the dataset at time $t_i$.
Then, items in $D(v)$ get deleted as time goes by.  

A simple solution for the deletion-only problem is to first release $\mathcal{M}(D(v))$ when $v_i$ is initialized, and then run an insertion-only mechanism over the deletions.  To answer a query $f$ at time $t$, we obtain $f(D_{t}^{\operatorname{del}}(v))$ from the insertion-only mechanism, where $D_{t}^{\operatorname{del}}(v)$ denotes the set of items in $D(v)$ that have been deleted by time $t$.  Then, we use $f(D_t(v)) = f(D(v)) - f(D_t^{\operatorname{del}}(v))$ as the answer, where $f(D(v))$ can be queried from $\mathcal{M}(D(v))$.  However, the error of this simple solution will be $\alpha(\abs{D(v)})$ from $\mathcal{M}(D(v))$.  On the other hand, our target error bound is $\alpha(n_t)$, but $\abs{D(v)}$ can be arbitrarily larger than $n_t$.   
To fix the problem, we ensure that no more than $\frac{\abs{D(v)}}{2}$ items should be deleted so that $|D(v)|=O(|D_t(v)|)$.  When half the items have been deleted from $D(v)$, we restart the process with a new $D(v)$ that consist of the remaining items. 

There are still a few privacy-related issues with the above idea.  First, we cannot restart when exactly half the items have been deleted, which would violate DP.  Instead, we run a basic counting mechanism over the deletions of $D(v)$ to approximately keep track of the number of deletions; we show that such an approximation will only contribute an additive polylogarithmic error.  Second, since in the online interval tree, each item has copies in multiple nodes of $\T$, and in each node, we restart the process above multiple times, we need to allocate the privacy budget carefully using sequential composition.  But the privacy degradation is only polylogarithmic since both numbers are logarithmic.  Finally, the first logarithm, i.e., the number of copies of each item, is the height of the tree $\log m_t$ where $m_t=N_t + \log\log t$, but $m_t$ is not known in advance.  Thus, we allocate a privacy budget proportional to  $1 / \ell^2$ to a node at level $\ell$ for $\ell=1,2,\dots,\log m_t$, so that the total privacy is bounded for any $m_t$, while incurring another logarithmic-factor degradation.  Note that we could have used a tighter series $1 / \ell^{1+\eta}$ for any constant $\eta>0$ but we did not try to optimize the polylogarithmic factors for the fully-dynamic algorithm.  Algorithm \ref{algo:node} details the steps we run at each node $v$ in the online interval tree.
We present in Lemma~\ref{lemma:node} its accuracy guarantee, assuming each node is allocated with $(\varepsilon,\delta)$-DP.

\begin{lemma}~\label{lemma:node}
For each node $v$, Algorithm~\ref{algo:node} is $(\varepsilon,\delta)$-DP. Suppose there is a static $(\varepsilon,\delta)$-DP mechanism for answering a class $\mathcal{Q}$ of linear queries with error function $\alpha^{(k)}(\varepsilon,\delta,\beta,\abs{D})$. Then Algorithm~\ref{algo:node} answers $\mathcal{Q}(D_{t}(v))$ for any node $v$ at time $t=t_j$ with error $O\left(\alpha^{(\log j)}\left(\tilde{O}(\varepsilon),\tilde{O}(\delta), \beta, n_{t}(v)+\tilde{O}(1)\right)\right)$ with probability at least $1-\beta$.
\end{lemma}

\begin{proof}
\textbf{Privacy.}
Algorithm~\ref{algo:node} uses four black-box mechanisms: When $v$ is initialized, the Laplace mechanism is used to protect its size $|D(v)|$, and a static mechanism $\mathcal{M}_{\mathcal{Q}}$ is used to answer $\mathcal{Q}(D(v))$; then two insertion-only mechanisms $\mathcal{M}_{\mathcal{Q}-\operatorname{Ins}}$ and $\mathcal{M}_{\operatorname{Lap-Ins}}$ are used to compute $\mathcal{Q}$ and the basic counting query respectively over the deletions.
In any round $r$, the composition of these four mechanisms is $(4\varepsilon_r, 3\delta_r)=\left(\frac{6\varepsilon}{\pi^2r^2}, \frac{6\delta}{\pi^2r^2}\right)$-DP.
As we restart these four mechanisms, they are sequentially composed, which guarantees the whole mechanism at node $v$ is $\left(\sum_{r=1}^\infty \frac{6\varepsilon}{\pi^2r^2}, \sum_{r=1}^\infty \frac{6\delta}{\pi^2r^2}\right)=(\varepsilon,\delta)$-DP, independent of the number of restarts.

\medskip
\textbf{Accuracy.}
We first bound the number of restarts $r$ as follows.
Note that $\gamma_{t,r}$ is the error bound of the Laplace mechanism at time $t$ in round $r$.
When a restart happens at time $t=t_j$, we have $\tilde{n}^{\operatorname{del}}_t > \tilde{n}/2 + 2\gamma_{t,r}$.
With probability $1-2\beta_r$, both $\tilde{n}^{\operatorname{del}}_t$ and $\tilde{n}$ have error at most $\gamma_{t,r}$. Conditioned on this happening, $n^{\operatorname{del}}_t>n(v)/2$: at least half of the remaining items have been deleted since the last restart.
Since $v$ was initialized with $n(v)\leq N_t$ items, this can happen at most $r=O(\log n(v))=O(\log N_t)$ times before there are only $n_t(v)\leq \gamma_{t,r}$ items left. Afterwards $\tilde{n}<2\gamma_{t,r}$ is true and the algorithm halts by answering $0$, which has error at most $n_t(v)=O(\gamma_{t,r})$.
Therefore, with probability $1-\sum_{r=1}^\infty 2\beta_r = 1-\beta/3$, there can only be $O(\log N_t)$ rounds.
We condition on this event in the following.

We next bound the error of $\mathcal{Q}(D_{t}(v))$ in any round $r=O(\log N_t)$.
If the algorithm does restart at time $t$, the latest dataset is computed and a fresh static mechanism $\mathcal{M}_\mathcal{Q}$ with privacy budget $(\varepsilon_r,\delta_r)$ is used to answers $\mathcal{Q}(D(v))$, whose error is $\alpha(\varepsilon_r,\delta_r, \beta/2, n_t(v))=O(\alpha(\varepsilon_r, \delta_r, \beta, n_t(v)))$ with probability $1-\beta/2$.
Otherwise (line~\ref{line:few}), with probability $1-2\beta_r$ we have the actual number of deletions
$n^{\operatorname{del}}_t\leq n(v)/2 + 4\gamma_{t,r}$, where $n(v)$ is the number of remaining items in the current round from the last restart.
Namely the current data size is at least $n_t(v)=n(v)-n^{\operatorname{del}}_t\geq n(v)/2 - 4\gamma_{t,r}$.
Further conditioned on this, we have with probability $1-\beta/6$, the error for $\mathcal{Q}(D(v))$ obtained from $\mathcal{M}(Q)$ is $\alpha(\varepsilon_r, \delta_r, \frac{\beta}{6}, n(v))$.
Also with probability $1-\beta/6$, the error of $\mathcal{Q}(D_t^{\operatorname{del}}(v))$ obtained from $\mathcal{M}_{\mathcal{Q}-\operatorname{Ins}}$ is 
$O\left(\alpha^{(\log j)}\left(\frac{\varepsilon_r}{\log j}, \frac{\delta_r}{\log j} , \frac{\beta}{6}, n^{\operatorname{del}}_t\right)\right)$ using Theorem~\ref{thm:our-ins}\footnote{Note that the additive $\frac{1}{\varepsilon}\log\frac{t}{\beta}$ term in Theorem~\ref{thm:our-ins} does not apply here, since we do not need to invoke another infinite partitioning.}.
Both terms are covered by $O\left(\alpha^{(\log j)}\left(\frac{\varepsilon_r}{\log j}, \frac{\delta_r}{\log j} , \beta, n_t(v)+\gamma_{t,r}\right)\right)$, as $n^{\operatorname{del}}_t+n(v)\leq 3n(v)/2+4\gamma_{t,r}\leq 3n_t(v)+16\gamma_{t,r}=O(n_t(v)+\gamma_{t,r})$. So at any time $t=t_j$ and in any round $r$, the error of answering $\mathcal{Q}(D_t(v))$ can be bounded by $O\left(\alpha^{(\log j)}\left(\frac{\varepsilon_r}{\log j}, \frac{\delta_r}{\log j} , \beta, n_t(v)+\gamma_{t,r}\right)\right)$ with probability $1-2\beta_r-2\beta/6\geq 1-2\beta/3$.

Finally for $t=t_j$, take $\gamma_{t,r}=O(\frac{1}{\varepsilon}\log^{1.5} j\log\frac{r^2}{\beta})$ from Theorem~\ref{thm:our-ins}, $(\varepsilon_r,\delta_r)= (\Theta(\frac{\varepsilon}{r^2}), \Theta(\frac{\delta}{r^2}))$, and $r=O(\log N_t)$, the error bound for $\mathcal{Q}(D_t(v))$ is (with probability $1-\beta$)
\[
O\left(\alpha^{(\log j)}\left(\frac{\varepsilon}{\log^2 N_t\log j}, \frac{\delta}{\log^2 N_t \log j} , \beta, n_{t}(v)+\frac{1}{\varepsilon}\log^{1.5} j\log\frac{\log N_t}{\beta}\right)\right)\,.
\]
\end{proof}

\subsection{Putting it Together}
\label{sec:together}
Lemma~\ref{lemma:node} assumes each node is under $(\varepsilon,\delta)$-DP, which we cannot afford since we have an outer tree of depth $\log m_t$.
Instead we allocate $(\varepsilon(v),\delta(v))=\left(\frac{6\varepsilon}{\pi^2\ell^{2}}, \frac{6\delta}{\pi^2\ell^{2}}\right)$ to a node $v$ at level $\ell$ in the outer tree, so that the composed mechanism is still $(\varepsilon,\delta)$-DP.
The final sum consists of a disjoint union of at most $\log m_t$ nodes, where the error of each node is given by Lemma~\ref{lemma:node} under $(\varepsilon(v),\delta(v))$-DP and $j\leq m_t$.
We have the total error at any time $t$ is 
\[
O\left(\alpha^{(\log^2 m_t)}\left(\frac{\varepsilon}{\log^2 N_t \log^{3} m_t},\frac{\delta}{\log^2 N_t \log^{3} m_t}, \beta, n_t+\frac{1}{\varepsilon}\log^{1.5}m_t\log\frac{\log m_t}{\beta}\right)\right)
\]
Finally, conditioned on $m_t=O(N_t+\log\log t)$ (which always holds for small $\beta$), and include the additive error from segmentation, we get the following theorem.
\begin{theorem}~\label{thm:dynamic}
Suppose there is a static $(\varepsilon,\delta)$-DP mechanism for answering a class $\mathcal{Q}$ of linear queries with error function $\alpha^{(k)}$.
There exists a mechanism for fully-dynamic streams such that at time $t$ it has error
\[
O\left(\alpha^{(\log^2 m_t)}\left(\frac{\varepsilon}{\log^{5}m_t},\frac{\delta}{\log^{5}m_t}, \beta, n_t+\frac{1}{\varepsilon}\log^{1.5} m_t\log\frac{\log m_t}{\beta}\right)+\frac{1}{\varepsilon}\log \frac{t}{\beta}\right),
\]
with probability $1-\beta$, where $m_t=N_t+\log\log t$, where $N_t$ is the number of updates and $n_t$ is the size of the dataset at time $t$. 
\end{theorem}
When instantiated with the PMW mechanism where $\alpha_{\operatorname{PMW}}^{(k)}=\tilde{O}(k\cdot \alpha_{\operatorname{PMW}})$,  the above bound simplifies to  $\tilde{O}(n_t^{1/2})$ (for $\delta>0$) or $\tilde{O}(n_t^{2/3})$ (for $\delta=0$), matching the optimal error  bound in the static setting up to polylogarithmic factors.

\bibliographystyle{unsrt}  
\bibliography{paper}  

\appendix
\section{DP Mechanisms for Linear Queries: A Review}
\label{app:review}

In this section, we present some important DP mechanisms for linear queries and their error bounds below.
We first analyze the error of a single mechanism, which is similar to the analysis in~\cite{DBLP:books/sp/17/Vadhan17}, where we clarify the dependency on $\beta$.
We then discuss the error bounds for a disjoint union of such mechanisms.

\subsection{Error Bounds}

\medskip
\noindent\textbf{Laplace Mechanism.}
When $\mathcal{Q}=\{f\}$ is a single query, the Laplace mechanism $\mathcal{M}_{\operatorname{Lap}}(D)=f(D)+\operatorname{Lap}(\frac{1}{\varepsilon})$ has error $\alpha_{\operatorname{Lap}}(\varepsilon,\beta)=\frac{1}{\varepsilon}\ln{\frac{1}{\beta}}$.
When $\mathcal{Q}$ contains multiple queries, we may add $\operatorname{Lap}(\frac{\abs{\mathcal{Q}}}{\varepsilon})$ noise to each query result and apply basic composition to guarantee $\varepsilon$-DP of the whole mechanism.
To translate it into an error bound, we bound the failure probability of each noise by $\frac{\beta}{\abs{\mathcal{Q}}}$, so that a union bound will bring the total failure probability to $\beta$.
A similar argument can be made using advanced composition.
To conclude, answering a set of queries $\mathcal{Q}$ using the Laplace mechanism achieves error (for $\delta\geq 0$)
\[
\alpha_{\operatorname{Lap}}(\varepsilon, \delta, \abs{\mathcal{Q}}, \beta)=\left\{\begin{aligned}
&O\left(\frac{\abs{\mathcal{Q}}}{\varepsilon}\log\frac{\abs{\mathcal{Q}}}{\beta}\right) & ,\,\delta\leq e^{-\Omega(\abs{\mathcal{Q}})}\,;\\
&O\left(\frac{\sqrt{\abs{\mathcal{Q}}\log\frac{1}{\delta}}}{\varepsilon}\log\frac{\abs{\mathcal{Q}}}{\beta}\right) & ,\,\delta\geq e^{-O(\abs{\mathcal{Q}})}\,.\\
\end{aligned}\right.
\]

\medskip
\noindent\textbf{Gaussian Mechanism.}
Similar to the Laplace mechanism, the Gaussian mechanism protects $(\varepsilon,\delta)$-DP of query $f$ by outputting $\mathcal{M}_{\operatorname{Gauss}}(D)=f(D)+\mathcal{N}\left(0, \frac{2}{\varepsilon^2}\ln\frac{1.25}{\delta}\right)$, and $\alpha_{\operatorname{Gauss}}(\varepsilon,\delta,\beta)=\frac{2}{\varepsilon}\sqrt{\ln\frac{1.25}{\delta}\ln\frac{2}{\beta}}$.
When composing multiple Gaussian mechanisms that each answers a query from $\mathcal{Q}$, zCDP composition~\cite{DBLP:conf/tcc/BunS16} can be applied, which shows adding $\mathcal{N}\left(0,O\left(\frac{\abs{\mathcal{Q}}}{\varepsilon^2}\log\frac{1}{\delta}\right)\right)$ noise to each query suffices to protect $(\varepsilon,\delta)$-DP of the whole mechanism.
Therefore the Gaussian mechanism achieves the following error (for $\delta > 0$) answering a set of queries $\mathcal{Q}$.
\[
\alpha_{\operatorname{Gauss}}(\varepsilon, \delta, \abs{\mathcal{Q}}, \beta)
=O\left(\frac{\sqrt{\abs{\mathcal{Q}}\log{\frac{1}{\delta}}}}{\varepsilon}\sqrt{\log\frac{\abs{\mathcal{Q}}}{\beta}}\right)\,.
\]

\medskip
\noindent\textbf{Private Multiplicative Weights.} 
When there are many queries $\abs{\mathcal{Q}}=\Omega(\abs{D})$, composing individual mechanisms has error polynomial in $\abs{\mathcal{Q}}$, thus also in $\abs{D}$.
The Private Multiplicative Weights mechanism~\cite{DBLP:conf/focs/HardtR10,DBLP:conf/nips/HardtLM12} performs better in this case. The following error bound is presented in \cite{DBLP:conf/nips/HardtLM12,DBLP:journals/fttcs/DworkR14}.
\[
\alpha_{\operatorname{PMW}}(\varepsilon,\delta,\beta,\abs{D},\abs{\mathcal{Q}},\abs{\mathcal{X}})=\left\{
\begin{aligned}
&O\left(\abs{D}^\frac{2}{3}\left(\frac{\log\abs{\mathcal{X}} \log \frac{\abs{\mathcal{Q}}}{\beta}}{\varepsilon}\right)^\frac{1}{3}\right), & \delta = 0; \\
&O\left(\abs{D}^\frac{1}{2}\left(\frac{\sqrt{\log \abs{\mathcal{X}}\log\frac{1}{\delta}} \log \frac{\abs{\mathcal{Q}}}{\beta}}{\varepsilon}\right)^\frac{1}{2}\right), & \delta > 0.
\end{aligned}
\right.
\]

\medskip
Apart from  mechanisms mentioned above, there are other private mechanisms for linear queries.
For example, the optimal composition \cite{DBLP:journals/tit/KairouzOV17} can be used in place of basic or advanced composition to provide a better allocation of privacy budget, yet computing it is costly.
The $\log \abs{\mathcal{Q}}$ factor is removable for Laplace mechanism~\cite{DBLP:journals/jpc/SteinkeU16} and almost removable for Gaussian mechanism~\cite{DBLP:conf/forc/GaneshZ21}.
Under pure-DP, SmallDB~\cite{DBLP:journals/jacm/BlumLR13} has asymptotically the same error as PMW, but its running time is prohibitive. The Matrix mechanism~\cite{DBLP:conf/pods/LiHRMM10,DBLP:journals/vldb/LiMHMR15} exploits structural properties within the query set $\mathcal{Q}$ and works well in practice. But it does not have a closed-form error bound for general queries.

In general, the best mechanism is related to the hereditary discrepancy~\cite{DBLP:conf/stoc/HardtT10,DBLP:conf/stoc/MuthukrishnanN12} of the set of queries.
For example, for $d$-dimensional halfspace counting queries, \cite{DBLP:conf/stoc/MuthukrishnanN12} has error $O(n^{\frac{1}{2} - \frac{1}{2d}}/\varepsilon)$ with high probability.
In this paper we use $\alpha$ as a function of $\varepsilon,\delta,\beta$, and possibly $\abs{D},\abs{\mathcal{X}}, \abs{\mathcal{Q}}$ to denote the error of any mechanism answering linear queries on \textit{static} datasets, without detailing the best mechanism under a specific setting and choice of the parameters. Since our paper takes a black-box approach, all these algorithms can be plugged into our framework so as to support dynamic data, while incurring a polylogarithmic-factor degradation. 

\subsection{Error Bounds under Disjoint Union}

In this paper we use $\alpha^{(k)}(\varepsilon,\delta,\beta, \abs{D},\dots)$ to denote the error (with probability $1-\beta$) of the sum of $k$ mechanisms on disjoint datasets. 
The $\varepsilon$ and $\delta$ here requires each individual mechanism to be $(\varepsilon,\delta)$-DP.
The privacy requirement of the whole mechanism is analyzed separately, and fed as input to the $\alpha^{(k)}$ function. The $\abs{D}$ here denotes the total size $\sum_{i=1}^k \abs{D_i}$ of the $k$  disjoint datasets.

It always holds that $\alpha^{(k)}(\varepsilon,\delta,\beta)\leq k\cdot \alpha(\varepsilon,\delta,\frac{\beta}{k})$ by taking a union bound over the guarantees $k$ individual mechanisms.
In this section, we show cases where $\alpha^{(k)}$ can be tighter for specific mechanisms.
We will use $\alpha_{\operatorname{Lap}}(\varepsilon,\beta)=\frac{1}{\varepsilon}\ln\frac{1}{\beta}$ as our running example.
We immediately have $\alpha_{\operatorname{Lap}}^{(k)}(\varepsilon,\beta)\leq \frac{k}{\varepsilon}\ln\frac{k}{\beta}$ by the union bound reduction.

\medskip
\noindent\textbf{Unbiasedness.}
If a mechanism $\mathcal{M}$ is \textit{unbiased} with error $\alpha(\varepsilon,\delta,\beta)$, naturally the error only grows by $\sqrt{k}$.
We can argue that with all but $\frac{\beta}{2}$ probability, each of the $k$ mechanisms simultaneously has its error bounded by $\alpha(\varepsilon,\delta,\frac{\beta}{2k})$. Conditioned on this happening, apply Hoeffding's inequality with the remaining $\frac{\beta}{2}$ probability, we get
\[
\alpha_{\operatorname{Unbiased}}^{(k)}(\varepsilon,\delta,\beta)\leq \sqrt{2k\ln\frac{4}{\beta}}\cdot \alpha(\varepsilon,\delta,\frac{\beta}{2k})\,.
\]
The unbiaseness saves a $\sqrt{k}$ dependency on $k$.
For the Laplace mechanism, this means
\[
\alpha_{\operatorname{Lap}}^{(k)}(\varepsilon,\beta)=
O\left(\frac{\sqrt{k}}{\varepsilon}\sqrt{\log\frac{1}{\beta}} \log\frac{k}{\beta}\right) \,.
\]

\medskip
\noindent\textbf{Concentration Bounds.}
For specific distributions like the Laplace (sub-exponential) and the Gaussian (sub-gaussian), concentration bounds are tighter than bounds derived by their unbiasedness.
In general it saves the $\log k$ factor from applying union bound.
For the Laplace mechanism, note that the $\operatorname{Lap}(\frac{1}{\varepsilon})$ random variable is sub-exponential with norm $\|\mathrm{Lap}(\frac{1}{\varepsilon})\|_{\Psi_1}=\frac{2}{\varepsilon}$.
We can then apply Bernstein’s inequality~\cite{book:HDP}.
\begin{lemma}[Bernstein's inequality]
Let $X_1,\dots,X_k$ be i.i.d. zero-mean sub-exponential random variables with norm $\Psi_1$. There is an absolute constant $c$ so that for any $t\geq 0$,
\[
\Pr\left[\left|\sum_{i=1}^k X_i\right| > t\right]\leq 2\exp\left[-c\min\left\{\frac{t^2}{k \Psi_1^2},\frac{t}{\Psi_1}\right\}\right]
\]
\end{lemma}
This gives a tighter error function for the Laplace mechanism
\[
\alpha_{\operatorname{Lap}}^{(k)}(\varepsilon,\beta)=O\left(\frac{\sqrt{k\log\frac{1}{\beta}}+\log\frac{1}{\beta}}{\varepsilon}\right)\,.
\]

To give another example, the sum of $k$ Gaussian noises is still a Gaussian noise with the variance scaled up by $k$, thus the disjoint union of $k$ Gaussian mechanisms has error function
\[
\alpha_{\operatorname{Gauss}}^{(k)}(\varepsilon,\delta,\beta)=O\left(\frac{\sqrt{k\log\frac{1}{\delta}\log\frac{1}{\beta}}}{\varepsilon}\right)\,.
\]

\medskip

\begin{algorithm}[ht]
\caption{$(\varepsilon,\delta)$-DP Algorithm at node $v=v_i$}\label{algo:node}
\KwIn{Fully-dynamic update stream $((x_1,c_1),\dots,(x_t,c_t),\dots)$, timestamps $(t_1,\dots, t_j,\dots)$, online interval tree node $v=v_i$, probability $\beta$,  privacy budget $(\varepsilon,\delta)$}
\KwIn{Static mechanism $\mathcal{M}_\mathcal{Q}$ and insertion-only mechanism $\mathcal{M}_{\mathcal{Q}-\operatorname{Ins}}$ for answering queries $\mathcal{Q}$, continual counting mechanism $\mathcal{M}_{\operatorname{Lap-Ins}}$}
\KwOut{$\mathcal{Q}(D_t(v))$ at any time $t=t_j$}
\tcc{Initialize}
$r\gets 1$, $(\varepsilon_r, \delta_r)\gets \left(\frac{3\varepsilon}{2\pi^2r^2},\frac{2\delta}{\pi^2r^2}\right), \beta_r\gets\frac{1}{\pi^2}\beta$\;
$D(v)\gets$ All items in $D_{t_i}$ inserted after the closest left-ancestor of $v$ in the online interval tree\;
$\tilde{n}\gets \abs{D(v)} + \operatorname{Lap}(\frac{1}{\varepsilon_r})$\;
Release $\mathcal{M}_{\mathcal{Q}}(D(v))$ under $(\varepsilon_r, \delta_r)$-DP as \textbf{output} to answer $\mathcal{Q}(D(v))$\;
Initiate $\mathcal{M}_{\mathcal{Q}-\operatorname{Ins}}$ and $\mathcal{M}_{\operatorname{Lap-Ins}}$, each under $(\varepsilon_r,\delta_r)$-DP\;
\tcc{Handle deletions}
\ForEach{$j\gets {i+1}, {i+2},\dots$} {
    \ForEach(\tcc*[f]{Update D(v)}){update $(x, c)$ in segment $s_j$}{
    \eIf{$c=-1$ and $x\in D(v)$}{
        Augment the deletion time of $x \in D(v)$ to $(x, t_j)$\;
        Feed an update $x$ to $\mathcal{M}_{\mathcal{Q}-\operatorname{Ins}}$ and $\mathcal{M}_{\operatorname{Lap-Ins}}$\;
    }{
        Feed an update $\bot$ to $\mathcal{M}_{\mathcal{Q}-\operatorname{Ins}}$ and $\mathcal{M}_{\operatorname{Lap-Ins}}$\;
    }
    }
$\tilde{n}_t^{\operatorname{del}}\gets $ The number of deleted items obtained from $\mathcal{M}_{\operatorname{Lap-Ins}}$\;
$\gamma_{t,r} \gets$ The (public) error bound of $\mathcal{M}_{\operatorname{Lap-Ins}}$ at time $t$ with probability $1-\beta_r$\;
\eIf(\tcc*[f]{Restart}){$\tilde{n}_t^{\operatorname{del}} > \tilde{n}/2 + 2\gamma_{t,r}$} {
    $r\gets r+1$, $(\varepsilon_r, \delta_r)\gets \left(\frac{3\varepsilon}{2\pi^2r^2},\frac{2\delta}{\pi^2r^2}\right),\beta_r\gets\frac{1}{\pi^2r^2}\beta$\; 
    $D(v)\gets D(v) -$ all augmented items in $D(v)$\;
    $\tilde{n}\gets \abs{D(v)}+\operatorname{Lap}(\frac{1}{\varepsilon_r})$\;
    \If(\tcc*[f]{Terminate}){$\tilde{n} < 2\gamma_{t,r}$}{
        Halt by answering 0 for all future $\mathcal{Q}(D_t(v))$\;
    } 
    Release a new $\mathcal{M}_{\mathcal{Q}}(D(v))$ under $(\varepsilon_r, \delta_r)$-DP as \textbf{output} to answer $\mathcal{Q}(D_{t_j}(v))$ \;
    Restart both $\mathcal{M}_{\mathcal{Q}-\operatorname{Ins}}$ and $\mathcal{M}_{\operatorname{Lap-Ins}}$, each under $(\varepsilon_r,\delta_r)$-DP\;
}{\label{line:few}
    Obtain the query result for the current round $\mathcal{Q}(D(v))$ from $\mathcal{M}_{\mathcal{Q}}$\;
    Obtain the query result for deleted items $\mathcal{Q}(D_t^{\operatorname{del}}(v))$ from $\mathcal{M}_{\mathcal{Q}-\operatorname{Ins}}$\;
    \textbf{Output} $\mathcal{Q}(D_{t}(v))$ as $\mathcal{Q}(D_{t}(v))\gets \mathcal{Q}(D(v))-\mathcal{Q}(D_{t}^{\operatorname{del}}(v))$\;\label{line:output}
}
}
\end{algorithm}

\end{document}

%% file: batch.tex
\begin{tikzpicture}[scale=0.5]
    
    \draw[|<->|][thick] (0,0) -- (5,0);
    \draw[<->|][thick] (5,0) -- (10,0);
    \draw[<->|][thick] (10,0) -- (12.5,0);
    \draw[<->|][thick] (12.5,0) -- (18,0); 
    \draw[-][thick] (18,0) -- (18.5,0); 
    
    \draw(0,-0.5) node[] {1};
    \draw(5,-0.5) node[] {$t_1$};
    \draw(10,-0.5) node[] {$t_2$};
    \draw(12.5,-0.5) node[] {$t_3$};
    \draw(18,-0.5) node[] {$t_4$};
    \draw(19.5,-0.5) node[] {$\cdots$};
    
    
    \draw (1,2) node[label=left:Input Updates] {};
    
    \draw[-] (3,3) -- (7,3); 
    \draw[fill] (3,3) circle (3pt);
    \draw[fill=white] (7,3) circle (3pt);
    \draw(5,3.5) node[] {$a$};

    \draw[-] (4,2) -- (11,2); 
    \draw[fill] (4,2) circle (3pt);
    \draw[fill=white] (11,2) circle (3pt);
    \draw(7.5,2.5) node[] {$b$};

    \draw[-] (10,3) -- (17,3); 
    \draw[fill] (10,3) circle (3pt);
    \draw[fill=white] (17,3) circle (3pt);
    \draw(13.5,3.5) node[] {$c$};

    \draw[-] (13,2) -- (16,2); 
    \draw[fill] (13,2) circle (3pt);
    \draw[fill=white] (16,2) circle (3pt);
    \draw(14.5,2.5) node[] {$e$};
    
    \draw[-] (12,1) -- (15,1); 
    \draw[fill] (12,1) circle (3pt);
    \draw[fill=white] (15,1) circle (3pt);
    \draw(13.5,1.5) node[] {$d$};

    
    \draw (1,-2) node[label=left:Batched Updates] {};
    
    \draw[-] (5,-2) -- (10,-2); 
    \draw[fill] (5,-2) circle (3pt);
    \draw[fill=white] (10,-2) circle (3pt);
    \draw(7.5,-1.5) node[] {$a$};

    \draw[-] (5,-3) -- (12.5,-3); 
    \draw[fill] (5,-3) circle (3pt);
    \draw[fill=white] (12.5,-3) circle (3pt);
    \draw(7.5,-2.5) node[] {$b$};

    \draw[-] (10,-4) -- (18,-4); 
    \draw[fill] (10,-4) circle (3pt);
    \draw[fill=white] (18,-4) circle (3pt);
    \draw(14,-3.5) node[] {$c$};
    
    \draw[-] (12.5,-2) -- (18,-2); 
    \draw[fill] (12.5,-2) circle (3pt);
    \draw[fill=white] (18,-2) circle (3pt);
    \draw(15,-1.5) node[] {$d$};
    
    
    \draw[loosely dotted] (5,-4.5) -- (5,4.5);
    \draw[loosely dotted] (10,-4.5) -- (10,4.5);
    \draw[loosely dotted] (12.5,-4.5) -- (12.5,4.5);
    \draw[loosely dotted] (18,-4.5) -- (18,4.5);

\end{tikzpicture}

%% file: int-tree.tex
\begin{tikzpicture}[scale=0.45]
    \coordinate (v1) at (0,1);
    \coordinate (v2) at (2,3);
    \coordinate (v3) at (4,1);
    \coordinate (v4) at (6,5);
    \coordinate (v5) at (8,1);
    \coordinate (v6) at (10,3);
    \coordinate (v7) at (12,1);
    \coordinate (v8) at (14,7);
    
    \draw[-] (v1) -- (v2) -- (v3);
    \draw[-] (v5) -- (v6) -- (v7);    
    \draw[-] (v2) -- (v4) -- (v6);
    \draw[-] (v4) -- (v8);
    \draw[dashed] (v8) -- (15,6.75);    

    \draw[-|][thick] (-2,0) -- (0,0);
    \draw[-|][thick] (0,0) -- (2,0);
    \draw[-|][thick] (2,0) -- (4,0);
    \draw[-|][thick] (4,0) -- (6,0);
    \draw[-|][thick] (6,0) -- (8,0); 
    \draw[-|][thick] (8,0) -- (10,0); 
    \draw[-|][thick] (10,0) -- (12,0);
    \draw[-|][thick] (12,0) -- (14,0);
    \draw[dashed][thick] (14,0) -- (15,0);

    \draw(-2,-0.5) node[] {$1$};
    \draw(0,-0.5) node[] {$t_1$};
    \draw(2,-0.5) node[] {$t_2$};
    \draw(4,-0.5) node[] {$t_3$};
    \draw(6,-0.5) node[] {$t_4$};
    \draw(8,-0.5) node[] {$t_5$};
    \draw(10,-0.5) node[] {$t_6$};
    \draw(12,-0.5) node[] {$t_7$};
    \draw(14,-0.5) node[] {$t_8$};

    \filldraw[color=black,fill=white](v1) circle (15pt) node[label=above:{$\{b\}$}]{$v_1$};
    \filldraw[color=black,fill=white](v2) circle (15pt) node[]{$v_2$};    
    \filldraw[color=black,fill=white](v3) circle (15pt) node[]{$v_3$};    
    \filldraw[color=black,fill=white](v4) circle (15pt) node[label=above:{$\{a,c,d\}$}]{$v_4$};    
    \filldraw[color=black,fill=white](v5) circle (15pt) node[]{$v_5$};
    \filldraw[color=black,fill=white](v6) circle (15pt) node[label=above:{$\{e,f\}$}]{$v_6$};    
    \filldraw[color=black,fill=white](v7) circle (15pt) node[]{$v_7$};    
    \filldraw[color=black,fill=white](v8) circle (15pt) node[]{$v_8$};     
    
        
    \draw[-] (0,-1.25) -- (14,-1.25); 
    \draw[fill] (0,-1.25) circle (3pt);
    \draw[fill=white] (14,-1.25) circle (3pt);
    \draw(7,-1) node[] {$a$};

    \draw[-] (0,-2.25) -- (2,-2.25); 
    \draw[fill] (0,-2.25) circle (3pt);
    \draw[fill=white] (2,-2.25) circle (3pt);
    \draw(1,-2) node[] {$b$};

    \draw[-] (4,-2.25) -- (12,-2.25); 
    \draw[fill] (4,-2.25) circle (3pt);
    \draw[fill=white] (12,-2.25) circle (3pt);
    \draw(8,-2) node[] {$c$};

    \draw[-] (6,-3.25) -- (8,-3.25); 
    \draw[fill] (6,-3.25) circle (3pt);
    \draw[fill=white] (8,-3.25) circle (3pt);
    \draw(7,-3) node[] {$d$};

    \draw[-] (10,-3.25) -- (14,-3.25); 
    \draw[fill] (10,-3.25) circle (3pt);
    \draw[fill=white] (14,-3.25) circle (3pt);
    \draw(12,-3) node[] {$f$};  

    \draw[-] (8,-4.25) -- (12,-4.25); 
    \draw[fill] (8,-4.25) circle (3pt);
    \draw[fill=white] (12,-4.25) circle (3pt);
    \draw(10,-4) node[] {$e$};    
\end{tikzpicture}

%% file: our-tree.tex
\begin{tikzpicture}[scale=0.45]
    \coordinate (v1) at (0,1);
    \coordinate (v2) at (2,3);
    \coordinate (v3) at (4,1);
    \coordinate (v4) at (6,5);
    \coordinate (v5) at (8,1);
    \coordinate (v6) at (10,3);
    \coordinate (v7) at (12,1);
    \coordinate (v8) at (14,7);
    
    \draw[-] (v1) -- (v2) -- (v3);
    \draw[-] (v5) -- (v6) -- (v7);    
    \draw[-] (v2) -- (v4) -- (v6);
    \draw[-] (v4) -- (v8);
    \draw[dashed] (v8) -- (15,6.75);    

    \draw[-|][thick] (-2,0) -- (0,0);
    \draw[-|][thick] (0,0) -- (2,0);
    \draw[-|][thick] (2,0) -- (4,0);
    \draw[-|][thick] (4,0) -- (6,0);
    \draw[-|][thick] (6,0) -- (8,0); 
    \draw[-|][thick] (8,0) -- (10,0); 
    \draw[-|][thick] (10,0) -- (12,0);
    \draw[-|][thick] (12,0) -- (14,0);
    \draw[dashed][thick] (14,0) -- (15,0);

    \draw(-2,-0.5) node[] {$1$};
    \draw(0,-0.5) node[] {$t_1$};
    \draw(2,-0.5) node[] {$t_2$};
    \draw(4,-0.5) node[] {$t_3$};
    \draw(6,-0.5) node[] {$t_4$};
    \draw(8,-0.5) node[] {$t_5$};
    \draw(10,-0.5) node[] {$t_6$};
    \draw(12,-0.5) node[] {$t_7$};
    \draw(14,-0.5) node[] {$t_8$};

    \filldraw[color=black,fill=white](v1) circle (15pt) node[label=above:{$\{a,b\}$}]{$v_1$};
    \filldraw[color=black,fill=white](v2) circle (15pt) node[label=above:{$\{a\}$}]{$v_2$};    
    \filldraw[color=black,fill=white](v3) circle (15pt) node[label=above:{$\{c\}$}]{$v_3$};    
    \filldraw[color=black,fill=white](v4) circle (15pt) node[label=above:{$\{a,c,d\}$}]{$v_4$};    
    \filldraw[color=black,fill=white](v5) circle (15pt) node[label=above:{$\{e\}$}]{$v_5$};
    \filldraw[color=black,fill=white](v6) circle (15pt) node[label=above:{$\{e,f\}$}]{$v_6$};    
    \filldraw[color=black,fill=white](v7) circle (15pt) node[]{$v_7$};    
    \filldraw[color=black,fill=white](v8) circle (15pt) node[]{$v_8$};     
    
        
    \draw[-] (0,-1.25) -- (14,-1.25); 
    \draw[fill] (0,-1.25) circle (3pt);
    \draw[fill=white] (14,-1.25) circle (3pt);
    \draw(7,-1) node[] {$a$};

    \draw[-] (0,-2.25) -- (2,-2.25); 
    \draw[fill] (0,-2.25) circle (3pt);
    \draw[fill=white] (2,-2.25) circle (3pt);
    \draw(1,-2) node[] {$b$};

    \draw[-] (4,-2.25) -- (12,-2.25); 
    \draw[fill] (4,-2.25) circle (3pt);
    \draw[fill=white] (12,-2.25) circle (3pt);
    \draw(8,-2) node[] {$c$};

    \draw[-] (6,-3.25) -- (8,-3.25); 
    \draw[fill] (6,-3.25) circle (3pt);
    \draw[fill=white] (8,-3.25) circle (3pt);
    \draw(7,-3) node[] {$d$};

    \draw[-] (10,-3.25) -- (14,-3.25); 
    \draw[fill] (10,-3.25) circle (3pt);
    \draw[fill=white] (14,-3.25) circle (3pt);
    \draw(12,-3) node[] {$f$};  

    \draw[-] (8,-4.25) -- (12,-4.25); 
    \draw[fill] (8,-4.25) circle (3pt);
    \draw[fill=white] (12,-4.25) circle (3pt);
    \draw(10,-4) node[] {$e$};    
    
\end{tikzpicture}

%% file: proof.tex
\begin{tikzpicture}[scale=0.4]
    \coordinate (v1) at (0,1);
    \coordinate (v2) at (2,3);
    \coordinate (v3) at (4,1);
    \coordinate (v4) at (6,5);
    \coordinate (v5) at (8,1);
    \coordinate (v6) at (10,3);
    \coordinate (v7) at (12,1);
    \coordinate (v8) at (14,7);
    \coordinate (vx) at (-2,8);    
    
    \draw[-] (v1) -- (v2) -- (v3);
    \draw[-] (v5) -- (v6) -- (v7);    
    \draw[-] (v2) -- (v4) -- (v6);
    \draw[-] (v4) -- (v8) -- (vx);
    
    \draw[-|][thick] (-3,0) -- (-2,0);        
    \draw[-|][thick] (-2,0) -- (0,0);    
    \draw[-|][thick] (0,0) -- (2,0);
    \draw[-|][thick] (2,0) -- (4,0);
    \draw[-|][thick] (4,0) -- (6,0);
    \draw[-|][thick] (6,0) -- (8,0); 
    \draw[-|][thick] (8,0) -- (10,0); 
    \draw[-|][thick] (10,0) -- (12,0);
    \draw[-|][thick] (12,0) -- (14,0);
    \draw[-][thick] (14,0) -- (15,0);

    \draw(0,-0.5) node[] {};
    \draw(2,-0.5) node[] {$t_i$};
    \draw(4,-0.5) node[] {};
    \draw(6,-0.5) node[] {};
    \draw(8,-0.5) node[] {};
    \draw(10,-0.5) node[] {};
    \draw(12,-0.5) node[] {$t_q$};
    \draw(14,-0.5) node[] {};

    \filldraw[color=black,fill=white](v1) circle (15pt) node[]{};
    \filldraw[color=black,fill=lightgray](v2) circle (15pt) node[]{$v_i$}; 
    \filldraw[color=black,fill=white](v3) circle (15pt) node[]{}; 
    \filldraw[ultra thick,color=black,fill=lightgray](v4) circle (15pt) node[]{$v$}; 
    \filldraw[color=black,fill=white](v5) circle (15pt) node[]{};
    \filldraw[ultra thick,color=black,fill=white](v6) circle (15pt) node[]{}; 
    \filldraw[ultra thick,color=black,fill=white](v7) circle (15pt) node[]{$v_q$}; 
    \filldraw[color=black,fill=lightgray](v8) circle (15pt) node[]{}; 
    \filldraw[ultra thick, color=black,fill=white](vx) circle (15pt) node[]{};     

    \filldraw[color=black,fill=lightgray] (-2,-2.5) circle (15pt) node[label=right:{\;Nodes storing $[t_i, t_j)$}]{};  
    \filldraw[ultra thick,color=black,fill=white] (8,-2.5) circle (15pt) node[label=right:{\;Nodes queried by $t_q$}]{};         
        
    \draw[-] (2,-1.25) -- (15,-1.25); 
    \draw[fill] (2,-1.25) circle (3pt);
    
\end{tikzpicture}

%% file: ms.bbl
\begin{thebibliography}{10}

\bibitem{DBLP:journals/fttcs/DworkR14}
Cynthia Dwork and Aaron Roth.
\newblock The algorithmic foundations of differential privacy.
\newblock {\em Found. Trends Theor. Comput. Sci.}, 9(3-4):211--407, 2014.

\bibitem{DBLP:conf/focs/HardtR10}
Moritz Hardt and Guy~N. Rothblum.
\newblock A multiplicative weights mechanism for privacy-preserving data
  analysis.
\newblock In {\em {IEEE} Symposium on Foundations of Computer Science, {FOCS}},
  pages 61--70. {IEEE} Computer Society, 2010.

\bibitem{DBLP:conf/nips/HardtLM12}
Moritz Hardt, Katrina Ligett, and Frank McSherry.
\newblock A simple and practical algorithm for differentially private data
  release.
\newblock In {\em Proc. Advances in Neural Information Processing Systems,
  {NeurIPS}}, pages 2348--2356, 2012.

\bibitem{DBLP:conf/stoc/DworkNPR10}
Cynthia Dwork, Moni Naor, Toniann Pitassi, and Guy~N. Rothblum.
\newblock Differential privacy under continual observation.
\newblock In {\em Proc. {ACM} Symposium on Theory of Computing, {STOC}}, pages
  715--724. {ACM}, 2010.

\bibitem{DBLP:conf/icalp/ChanSS10}
T.{-}H.~Hubert Chan, Elaine Shi, and Dawn Song.
\newblock Private and continual release of statistics.
\newblock In {\em Proc. Automata, Languages and Programming, {ICALP}}, volume
  6199, pages 405--417. Springer, 2010.

\bibitem{DBLP:conf/asiacrypt/DworkNRR15}
Cynthia Dwork, Moni Naor, Omer Reingold, and Guy~N. Rothblum.
\newblock Pure differential privacy for rectangle queries via private
  partitions.
\newblock In {\em Proc. Advances in Cryptology, {ASIACRYPT}}, volume 9453,
  pages 735--751. Springer, 2015.

\bibitem{DBLP:books/sp/17/Vadhan17}
Salil~P. Vadhan.
\newblock The complexity of differential privacy.
\newblock In {\em Tutorials on the Foundations of Cryptography}, pages
  347--450. Springer International Publishing, 2017.

\bibitem{book:HDP}
Roman Vershynin.
\newblock {\em High-Dimensional Probability: An Introduction with Applications
  in Data Science}.
\newblock Cambridge University Press, 2018.

\bibitem{DBLP:conf/nips/CummingsKLT18}
Rachel Cummings, Sara Krehbiel, Kevin~A. Lai, and Uthaipon~Tao Tantipongpipat.
\newblock Differential privacy for growing databases.
\newblock In {\em Proc. Advances in Neural Information Processing Systems,
  {NeurIPS}}, pages 8878--8887, 2018.

\bibitem{bun2015differentially}
Mark Bun, Kobbi Nissim, Uri Stemmer, and Salil Vadhan.
\newblock Differentially private release and learning of threshold functions.
\newblock In {\em IEEE 56th Annual Symposium on Foundations of Computer
  Science}, pages 634--649, 2015.

\bibitem{DBLP:conf/stoc/MuthukrishnanN12}
S.~Muthukrishnan and Aleksandar Nikolov.
\newblock Optimal private halfspace counting via discrepancy.
\newblock In {\em Proc. Symposium on Theory of Computing Conference, {STOC}},
  pages 1285--1292. {ACM}, 2012.

\bibitem{DBLP:conf/focs/DworkRV10}
Cynthia Dwork, Guy~N. Rothblum, and Salil~P. Vadhan.
\newblock Boosting and differential privacy.
\newblock In {\em {IEEE} Symposium on Foundations of Computer Science, {FOCS}},
  pages 51--60. {IEEE} Computer Society, 2010.

\bibitem{DBLP:journals/tit/KairouzOV17}
Peter Kairouz, Sewoong Oh, and Pramod Viswanath.
\newblock The composition theorem for differential privacy.
\newblock {\em {IEEE} Trans. Inf. Theory}, 63(6):4037--4049, 2017.

\bibitem{DBLP:conf/tcc/BunS16}
Mark Bun and Thomas Steinke.
\newblock Concentrated differential privacy: Simplifications, extensions, and
  lower bounds.
\newblock In {\em Proc. Theory of Cryptography, {TCC}}, volume 9985, pages
  635--658, 2016.

\bibitem{DBLP:conf/sigmod/McSherry09}
Frank McSherry.
\newblock Privacy integrated queries: an extensible platform for
  privacy-preserving data analysis.
\newblock In {\em Proc. {ACM} {SIGMOD} International Conference on Management
  of Data, {SIGMOD}}, pages 19--30. {ACM}, 2009.

\bibitem{DBLP:conf/soda/HenzingerUU23}
Monika Henzinger, Jalaj Upadhyay, and Sarvagya Upadhyay.
\newblock Almost tight error bounds on differentially private continual
  counting.
\newblock In {\em Proc. {ACM-SIAM} Symposium on Discrete Algorithms, {SODA}},
  pages 5003--5039. {SIAM}, 2023.

\bibitem{DBLP:conf/icml/Fichtenberger23}
Hendrik Fichtenberger, Monika Henzinger, and Jalaj Upadhyay.
\newblock Constant matters: Fine-grained complexity of differentially private
  continual observation using completely bounded norms.
\newblock In {\em Proc. International Conference on Machine Learning, {ICML}},
  2023.

\bibitem{DBLP:journals/vldb/LiMHMR15}
Chao Li, Gerome Miklau, Michael Hay, Andrew McGregor, and Vibhor Rastogi.
\newblock The matrix mechanism: optimizing linear counting queries under
  differential privacy.
\newblock {\em {VLDB} J.}, 24(6):757--781, 2015.

\bibitem{DBLP:journals/pvldb/LyuSL17}
Min Lyu, Dong Su, and Ninghui Li.
\newblock Understanding the sparse vector technique for differential privacy.
\newblock {\em Proc. {VLDB} Endow.}, 10(6):637--648, 2017.

\bibitem{DBLP:conf/pods/DongY23}
Wei Dong and Ke~Yi.
\newblock Universal private estimators.
\newblock In {\em Proc. {ACM} {SIGMOD-SIGACT-SIGAI} Symposium on Principles of
  Database Systems, {PODS}}, pages 195--206. {ACM}, 2023.

\bibitem{DBLP:books/lib/BergCKO08}
Mark de~Berg, Otfried Cheong, Marc~J. van Kreveld, and Mark~H. Overmars.
\newblock {\em Computational geometry: algorithms and applications, 3rd
  Edition}.
\newblock Springer, 2008.

\bibitem{DBLP:journals/jpc/SteinkeU16}
Thomas Steinke and Jonathan~R. Ullman.
\newblock Between pure and approximate differential privacy.
\newblock {\em J. Priv. Confidentiality}, 7(2), 2016.

\bibitem{DBLP:conf/forc/GaneshZ21}
Arun Ganesh and Jiazheng Zhao.
\newblock Privately answering counting queries with generalized gaussian
  mechanisms.
\newblock In {\em Symposium on Foundations of Responsible Computing, {FORC}},
  volume 192 of {\em LIPIcs}, pages 1:1--1:18, 2021.

\bibitem{DBLP:journals/jacm/BlumLR13}
Avrim Blum, Katrina Ligett, and Aaron Roth.
\newblock A learning theory approach to noninteractive database privacy.
\newblock {\em J. {ACM}}, 60(2):12:1--12:25, 2013.

\bibitem{DBLP:conf/pods/LiHRMM10}
Chao Li, Michael Hay, Vibhor Rastogi, Gerome Miklau, and Andrew McGregor.
\newblock Optimizing linear counting queries under differential privacy.
\newblock In {\em Proc. {ACM} {SIGMOD-SIGACT-SIGART} Symposium on Principles of
  Database Systems, {PODS}}, pages 123--134. {ACM}, 2010.

\bibitem{DBLP:conf/stoc/HardtT10}
Moritz Hardt and Kunal Talwar.
\newblock On the geometry of differential privacy.
\newblock In {\em Proc. {ACM} Symposium on Theory of Computing, {STOC}}, pages
  705--714. {ACM}, 2010.

\end{thebibliography}
